\documentclass[]{article}
\makeatletter\if@twocolumn\PassOptionsToPackage{switch}{lineno}\else\fi\makeatother
\makeatletter
\usepackage{wrapfig}

\makeatother  
\usepackage{amsmath,amsfonts,amsbsy,amssymb,tabulary,graphicx,times,caption,fancyhdr}
\usepackage[utf8]{inputenc}
\usepackage[paperheight=10in,paperwidth=6.5in,margin=2cm,headsep=.5cm,top=2.5cm]{geometry}
\renewenvironment{abstract} {\vspace*{-1pc}\trivlist\item[]\leftskip\oupIndent\hrulefill\par\vskip4pt\noindent\textbf{\abstractname}\mbox{\null}\\}{\par\noindent\hrulefill\endtrivlist} 
 \date{}
\makeatletter\def\oupIndent{1pt}
\def\author#1{\gdef\@author{\hskip-\dimexpr(\tabcolsep)\hskip\oupIndent\parbox{\dimexpr\textwidth-\oupIndent}{\centering\bfseries#1}}}
\def\title#1{\gdef\@title{\centering\bfseries\ifx\@articleType\@empty\else\@articleType\\\fi#1}}
\let\@articleType\@empty \def\articletype#1{\gdef\@articleType{{\normalfont\itshape#1}}}
\fancypagestyle{headings}{\fancyhf{}\fancyhead[C]{\RunningHead}\fancyhead[R]{\thepage}}
\makeatletter
\AtBeginDocument{\@ifpackageloaded{textcomp}{}{\usepackage{textcomp}}}
\makeatother
\usepackage{colortbl}
\usepackage{xcolor}
\usepackage{pifont}
\usepackage[nointegrals]{wasysym}
\def\processVert{\ifmmode|\else\textbar\fi}

\makeatother

\usepackage[numbers,sort&compress]{natbib}

\usepackage[T1]{fontenc}
\usepackage{amsthm}
\newtheorem{fact}{Fact}
\newtheorem{theorem}{Theorem}

\newtheorem{corollary}{Corollary}
\newtheorem{definition}{Definition}
\newtheorem{example}{Example}

\newtheorem{remark}{Remark}
\usepackage[all]{xy}
\usepackage{tikz}
\newcommand{\lr}[1]{\langle #1 \rangle}
\allowdisplaybreaks[4]

\begin{document}
\title{Losing Connection: the Modal Logic of Definable Link Deletion
\medskip}

\author{\textbf{\fontsize{14pt}{16.4pt}\selectfont{Dazhu Li}}~
\medskip
\\\normalsize\normalfont {Department of Philosophy\unskip, Tsinghua University, Beijing \\ Institute for Logic, Language and Computation\unskip, University of Amsterdam, P.O. Box 94242, 1090 GE Amsterdam,
The Netherlands}~\\{\normalsize\normalfont  E-mail: lidazhu91@163.com}}
\def\RunningHead{{Losing Connection: the Modal Logic of Definable Link Deletion}}
\maketitle 

\begin{abstract}
In this article, we start with a two-player game that models communication under adverse circumstances in everyday life and study it from the perspective of a modal logic of graphs, where links can be deleted locally according to  definitions available to the adversarial player. We first introduce a new language, semantics, and some typical validities. We then formulate a new type of first-order translation for this modal logic and prove its correctness. Then, a novel notion of bisimulation is proposed which leads to a characterization theorem for the logic as a fragment of first-order logic, and a further investigation is made of its expressive power against hybrid modal languages. Next, we discuss how to axiomatize this logic of link deletion, using dynamic-epistemic logics as a contrast. Finally, we show that our new modal logic lacks both the tree model property and the finite model property, and that its satisfiability problem is undecidable. 
\end{abstract}

\smallskip\noindent\textbf{Keywords: }{Graph Game, Modal Logic, Dynamic Logic, Link Deletion, Undecidability}
  
\section{Introduction}\label{sec-introduction}
In the graph of the World-Wide Web, to search for relevant and valuable information, a computer user usually clicks through consecutive hyperlinks passing through intermediate web pages. However, hyperlinks do not always work: say, because of technical malfunctions, or more interestingly, intentional obstruction. Such scenarios of search under adverse circumstances are quite common, and formally, they can be modeled as non-cooperative games played on graphs. For instance, consider the following web graph:

\medskip
\smallskip

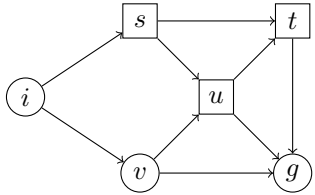
\begin{center}
\begin{tikzpicture}[every circle node/.style={draw,inner sep=0pt,minimum size=5mm}, every rectangle node/.style= {draw,inner sep=0pt,minimum size=4.5mm}]
\node(a)[circle] at (-0.5,0) {$i$};
\node(b) at (1,1) {$s$};
\node(c)[circle] at (1,-1) {$v$};
\node(d) at (2,0) {$u$};
\node(e) at (3,1) {$t$};
\node(f)[circle] at (3,-1) {$g$};
\draw[->](a) to (b);
\draw[->](a) to (c);
\draw[->](b) to (d);
\draw[->](b) to (e);
\draw[->](c) to (d);
\draw[->](c) to (f);
\draw[->](d) to (e);
\draw[->](d) to (f);
\draw[->](e) to (f);
\end{tikzpicture} 
\end{center}
In this picture, nodes stand for web pages, directed arrows are hyperlinks, and the two kinds of shape, square and circle, denote two different properties of web pages. One player $E$, the user in the above scenarios, starts at point $i$, and tries to arrive at one of the goal points $t$ and $g$. The other player $A$, say, Nature or some intentional opponent, tries to prevent this. The game goes in rounds: $A$ first cuts one or more links in the graph, then $E$ makes a step along some still available link. Since $A$ can cut at most 8 links in all, the game is finite. $E$ wins if she gets to the goal region, and loses if she cannot get there.

This description still leaves the game underspecified, since we must say more about how player $A$ is allowed to cut before we can analyze the outcomes of the game. For concreteness, we start with a variant where the properties are not yet essential.

\medskip

\noindent \textit{\bf First Version} \, Player $A$ cuts one arrow from $E$'s current position to some reachable node.

\medskip

In the resulting game on our graph, player $E$ has a winning strategy: she is always able to find the information that she needs. Player $A$ might start by deleting the link $\lr{i,s}$, then $E$ moves to node $v$. In the second round, $A$ must cut $\lr{v,g}$, and $E$ goes to state $u$. Finally, player $E$ can always arrive at $t$ or $g$ whatever link $A$ deletes.

In this first version, the game is a local variant of the \textit{sabotage game} (\textit{SG}) in \cite{lig}. A sabotage game is played on a graph by two players: in each round, \textit{Traveler} acts in the same way as $E$, while $A$'s counterpart \textit{Demon} first cuts a link. However, \textit{Demon}'s moves in sabotage games are global and allow cutting a link anywhere in the graph, not necessarily starting at the current position of \textit{Traveler}. In contrast, our game restricts the moves available to \textit{Demon}, giving him fewer winning strategies in general (cf. \cite{sabotage}).

\smallskip

However, the real-world scenarios that we considered suggest a more drastic deviation from existing sabotage games. In many cases of obstruction, the hostile opponent can cut more than one link, following a recipe rather than some arbitrary choice. For instance, blocking of links between computers is usually done by a program working on some explicit description of the targets to be blocked. Or for another concrete illustration of locality and definability, agents in a social network can cut friendship links starting with themselves, and they will often do that cutting according to some rule, such as `delete all links to people that have proved to be dishonest'.

Our next game models such more realistic scenarios, taking care of both aspects.

\medskip

\noindent\textit{\bf Definitive Version} \, In each round, player $A$ chooses an available atomic property, and cuts all links from the position of $E$ to nodes with the chosen property. 

\medskip

For example, in the above graph, when $E$ is located at node $s$, $A$ can cut both the links $\lr{s,u}$ and $\lr{s,t}$ if he chooses the definable property of nodes marked by the square. 

Clearly, with this new version, $A$'s powers of blocking access to information have increased. Indeed, on the same graph as before, he now has a winning strategy. In the first round, $A$ cuts the link $\lr{i,v}$, and $E$'s only option is to move to node $s$. But then, $A$ can cut both links $\lr{s,u}$ and $\lr{s,t}$ simultaneously, and $E$ gets stuck and loses.


We will now focus on the logical analysis of our second more realistic game,  calling it the \textit{definable sabotage game}, denoted S$_d$G. Here existing modal logics for sabotage can serve as an inspiration, given the similarity of the games. But they must be modified, since we have made the obstructing player both less powerful (given the local nature of his choices) and more powerful (since he can remove more than one link in general). More concretely, to analyze the sabotage game, \cite{sabotage} presents a \textit{sabotage modal logic} (\textit{SML}) extending standard modal logic with a modality $\blacklozenge\varphi$ stating that $\varphi$ is true at the evaluation point after removing some accessibility arrow from the model. But what is a suitable logic for S$_d$G? The next section contains our proposal, called \textit{definable sabotage modal logic} (\textit{S$_d$ML}). We will study this logic in depth, not just for its connections to the above games, but also as a pilot study for throwing light on what is special and what is general about sabotage games, and the logical theory that already exists for them. In addition, our logic is a test case for how local sabotage, even though definable in ways reminiscent of dynamic-epistemic logics of information update, has its own behavior, including significantly higher complexity (cf. \cite{johandynamic}).

%

%

\medskip

\noindent\textbf{Outline of the Paper.} In Section \ref{sec-languagesemanticsvalidities}, we present the syntax and semantics of S$_d$ML (Section \ref{subsec-languagandsemantics}), and some typical logical validities (Section \ref{subsec-validities}). In Section \ref{sec-translation}, we describe the non-trivial first-order translation for S$_d$ML and check its correctness. In Section \ref{sec-bisimulationcharacterization}, we first introduce a notion of bisimulation for S$_d$ML and investigate some of its model theory (Section \ref{subsec-bisimulation}), then we prove a characterization theorem for S$_d$ML as a fragment of first-order logic that is invariant for the bisimulation introduced (Section \ref{subsec-characterization}), and finally we explore the expressive power of S$_d$ML (Section \ref{subsec-compare}). In Section \ref{sec-hybrid}, we provide some further analysis of an axiomatization of S$_d$ML. In particular, we illustrate the relation between S$_d$ML and hybrid logics (Section \ref{subsec-hybridtranslation}), and study recursion axioms (Section \ref{subsec-recursion}). Next, in Section \ref{sec-undecidability}, we show that S$_d$ML lacks both the tree model property and the finite model property, and that the satisfiability problem for S$_d$ML is undecidable. Finally, we discuss related work in Section \ref{sec-relatedwork}, and conclude in Section \ref{sec-summary} with a summary and outlook.  

\section{Language, Semantics and Logical Validities}\label{sec-languagesemanticsvalidities}
In this section, we introduce the syntax and semantics of S$_d$ML. After that, to understand the new device, we illustrate some properties of the logic by means of logical validities.

\subsection{Language and Semantics}\label{subsec-languagandsemantics}
As mentioned above, the definable sabotage modal logic S$_d$ML is intended to match S$_d$G. Therefore its language should be expressive enough to model the actions of the players. For player $E$, it is natural to think of the standard modality $\Diamond$, which characterizes the transition from a node to its successors (see \cite{modallogic}). However, to characterize the action of $A$, some dynamic operator is indispensable. 

The language $\mathcal{L}_d$ of S$_d$ML is a straightforward extension of the standard modal language $\mathcal{L}_\Box$. In addition to the modality $\Diamond$, it also includes a dynamic modal operator $[-\;]$. The formal definition is as follows:

\begin{definition}[Language]\label{def-syntax}
Let {\rm{\textbf{P}}} be a countable set of propositional atoms. The formulas of $\mathcal{L}_d$ are defined by the following grammar in Backus-Naur Form:
$$\mathcal{L}_d\ni\varphi::=p\mid\neg\varphi\mid(\varphi\land\varphi)\mid\Box\varphi\mid[-\varphi]\varphi$$
where $p\in$\;{\rm{\textbf{P}}}. Besides, notions $\top$, $\bot$, $\lor$, $\to$ and $\Diamond$ are as usual. For any $[-\varphi]\psi\in\mathcal{L}_d$, we define $\lr{-\varphi}\psi:=\neg[-\varphi]\neg\psi$, i.e., $\lr{-\;}$ is the dual operator of $[-\;]$.
\end{definition}

We will often omit parentheses when doing so ought not cause confusion. The operator $[-\;]$ is our device to model the action of $A$ in S$_d$G. This can be clarified by the semantics of S$_d$ML. Formulas of $\mathcal{L}_d$ are evaluated in standard relational models $\mathcal{M}=\lr{W,R,V}$, where $W$ is the domain, a non-empty set of states, nodes or points, $R\subseteq W^2$ is the set of accessibility relations or links between points, and $V: \mathbf{P}\to2^W$ is the valuation function. A pair $\mathcal{F}=\lr{W,R}$ is called a \textit{frame}. For each $w\in W$, $\lr{\mathcal{M},w}$ is a \textit{pointed model}. For brevity, we usually write $\mathcal{M},w$ instead of $\lr{\mathcal{M},w}$. For any $\lr{w,v}\in R$, we also write $\lr{w,v}\in\mathcal{M}$. Besides, we use $R(w)$ to denote the set $\{v\in W\mid\lr{w,v}\in R\}$ of successors of $w$. We now introduce the semantics, which is defined inductively by truth conditions.

\begin{definition}[Semantics]\label{def-semantics}
Given a pointed model $\lr{\mathcal{M},w}$ and  a formula $\varphi$ of $\mathcal{L}_d$, we say that $\varphi$ is true in $\mathcal{M}$ at $w$, written as $\mathcal{M},w\vDash\varphi$, when 
\begin{align*}
\mathcal{M},w\vDash p&\;\;{\rm{\textit{iff}}}\;\; w\in V(p) \\
\mathcal{M},w\vDash\neg\varphi& \;\;{\rm{\textit{iff}}}\;\;\mathcal{M},w \not\vDash\varphi\\
\mathcal{M},w\vDash\varphi\land\psi &\;\;{\rm{\textit{iff}}}\;\;\mathcal{M},w\vDash \varphi\;and\;\mathcal{M},w\vDash\psi\\
\mathcal{M},w\vDash\Box\varphi& \;\;{\rm{\textit{iff}}}\;\; for\; each\; v\in W, \; if\; Rwv,\; then\;\mathcal{M},v\vDash\varphi\\
\mathcal{M},w\vDash[-\varphi]\psi& \;\;{\rm{\textit{iff}}}\;\; \mathcal{M}|_{\lr{w,\varphi}},w\vDash\psi
\end{align*}
where $\mathcal{M}|_{\lr{w,\varphi}}=\lr{W,R\setminus (\{w\}\times V(\varphi)\cap R(w)),V}=\lr{W,R\setminus (\{w\}\times V(\varphi)),V}$ is obtained by deleting all links from $w$ to the nodes that are $\varphi$.
\end{definition}

We say that formula $\varphi$ is \textit{satisfiable} if there exists a pointed model $\lr{\mathcal{M},w}$ such that $\mathcal{M},w\vDash\varphi$. By Definition \ref{def-semantics}, the truth conditions for Boolean and modal connectives $\neg$, $\land$, $\Box$ are as usual, and $[-\varphi]\psi$ means that $\psi$ is true at the evaluation point after deleting all accessibility relations from the current point to the nodes that are $\varphi$. Besides, we say that two pointed models $\lr{\mathcal{M}_1,w}$ and $\lr{\mathcal{M}_2,v}$ are \textit{$\varphi$-sabotage-related} (notation, $\lr{\mathcal{M}_1,w}\xrightarrow[]{-\varphi}\lr{\mathcal{M}_2,v}$) iff $\lr{\mathcal{M}_2,v}$ is $\lr{\mathcal{M}_1|_{\lr{w,\varphi}},w}$. Intuitively, by the semantics, formula $\varphi$ occurring in $[-\;]$ stands for a property of some successors of the current point, and $[-\varphi]$ is exactly an action of player $A$ in S$_d$G. 

\medskip

\noindent\textbf{Example Revisited.} Recall the graph at the outset. Assume that the propositional atoms $p$ and $q$ refer to the properties denoted with circle and square respectively. Then we are able to express the facts of the game with formulas of $\mathcal{L}_d$. For instance, that `after $A$ deletes the links from $v$ to the circle point, i.e., $g$, $E$ still can move to a square node, i.e., $u$' can be expressed as the truth at $v$ of the formula $[-p]\Diamond q$. Besides, $\mathcal{L}_d$ can also define the existence of winning strategies for players. For example, the formula $[-p]\Box[-q]\Box\bot$ states that $A$ can stop $E$ successfully by removing the links from the position of $E$ to the circle nodes in the first round, and cutting the links pointing to the square nodes in the second round. By our semantics for these formulas, S$_d$ML captures S$_d$G precisely.

\subsection{Logical Validities}\label{subsec-validities}
Although the language and semantics of S$_d$ML look simple, there are some issues with the new operator $[-\;]$. To illustrate how it works, we explore some interesting validities of S$_d$ML. First of all, let us consider the following principle:
\begin{eqnarray}
&&[-\varphi](\varphi_1\to\varphi_2)\to([-\varphi]\varphi_1\to[-\varphi]\varphi_2)
\end{eqnarray}
which follows from the semantics of S$_d$ML directly. The formula enables us to distribute $[-\;]$ over an implication. It is a common principle that applies to almost all modalities, e.g. the standard modality and the public announcement operator (see, e.g. \cite{baltag}). However, operator $[-\;]$ also has some distinguishing features. For instance, the validity
\begin{eqnarray}
&&[-\varphi]\psi\leftrightarrow\lr{-\varphi}\psi
\end{eqnarray}
illustrates that $[-\;]$ is self-dual and---less obviously---a model update function essentially. It is not hard to check that the validity of formulas (1) and (2) is closed under substitution. Interestingly, this is not a common feature of S$_d$ML. Some examples are as follows:
\begin{eqnarray}
&&[-\varphi]p\leftrightarrow p\\
&&[-p]\Diamond q\leftrightarrow\Diamond(\neg p\land q)\\
&&[-p][-q]\varphi\leftrightarrow[-q][-p]\varphi
\end{eqnarray}
Principle (3) illustrates that operator $[-\;]$ does not change the truth value of propositional atoms. Formula (4) allows us to reduce a formula including $[-\;]$ to an $\mathcal{L}_\Box$-formula. By (5), when all formulas occurring in $[-\;]$ are propositional atoms, the order of different operators $[-\;]$ can be interchanged.

Actually each propositional atom occurring in formulas (3)-(5) can be replaced by any Boolean formula without affecting their validity. However, these schematic validities fail in general when we consider the deletions for complex properties. As an example, we show this phenomenon for principle (5).

\begin{example}\label{example-validity}
Consider the general schematic form $[-\varphi_1][-\varphi_2]\varphi\leftrightarrow[-\varphi_2][-\varphi_1]\varphi$ for the principle (5). Let $\varphi_1:=p$, $\varphi_2:=\Diamond\Diamond p$, and $\varphi:=\Diamond q$. 
Define a model $\mathcal{M}$ as follows:
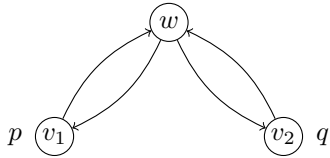
\begin{center}
 \begin{tikzpicture}[every node/.style={circle,draw,inner sep=0pt,minimum size=5mm}]
 \node(a) at (0,0){$w$};
 \node(b) [label=left:$p$]at (-1.5,-1.5){$v_1$};
 \node(c) [label=right:$q$]at (1.5,-1.5){$v_2$};
 \draw[->](a)to [bend left=20](b);
 \draw[->](a)to [bend right=20](c); 
 \draw[->](b)to [bend left=20](a);
 \draw[->](c)to [bend right=20](a); 
 \end{tikzpicture} 
 \end{center}
By inspection, one sees that $\mathcal{M},w\vDash[-p][-\Diamond\Diamond p]\Diamond q$ and $\mathcal{M},w\not\vDash[-\Diamond\Diamond p][-p]\Diamond q$. Therefore it holds that $\mathcal{M},w\not\vDash[-p][-\Diamond\Diamond p]\Diamond q\leftrightarrow[-\Diamond\Diamond p][-p]\Diamond q$.
\end{example}

Many instances of validity in S$_d$ML are not straightforward, and require much more thought than the often rather obvious validities found in standard logical systems. In particular, the dynamic modality $[-\;]$ creates interesting complexity, since removing a link in a model can have side-effects for truth values of formulas at worlds throughout the model. Therefore, it is time to make a deeper technical investigation of our logic. 

\section{First-Order Translation for S$_d$ML}\label{sec-translation}

Given the semantics of S$_d$ML, a natural question is: is S$_d$ML axiomatizable? Obviously the truth conditions for S$_d$ML are first-order, so there must be a first-order translation like that for standard modal logic. In this section, we present a positive answer to the question by describing a recursive standard translation for S$_d$ML. 

However we already know from SML that additional arguments may be needed in the translation: for SML, that extra argument was a finite set of links (see \cite{sabotage}). Interestingly, finding the translation here requires even more delicate analysis of the extra argument.

To do so, our method is to introduce a new device, being a sequence consisting of ordered pairs, e.g. $\lr{v,\varphi}$, to denote the occurrences of $[-\;]$ in a formula, where $v$ is a variable and $\varphi$ is a property of its successors. Let $\mathcal{L}_1$ be the first-order language consisting of countable unary predicates $P_{i\in N}$, a binary relation $R$, and equivalence $\equiv$. 

\begin{definition}[Standard Translation for S$_d$ML] \label{def-translation} 
Let $x$ be a designated variable, and $O$ be a finite sequence $\lr{v_0,\psi_0};...;\lr{v_i,\psi_i};...;\lr{v_n,\psi_n} (0\leqslant i\leqslant n)$, where $\psi_{0\leqslant i\leqslant n}$ is an $\mathcal{L}_d$-formula and $v_{0\leqslant i\leqslant n}$ is a variable. Then the translation $ST_x^O:\mathcal{L}_d\rightarrow\mathcal{L}_1$ is defined recursively as follows:
\begin{align*}
ST^O_x(p)&=Px\\
ST^O_x(\top)&=x\equiv x\\
ST^O_x(\neg\varphi)&=\neg ST^O_x(\varphi)\\
ST^O_x(\varphi_1\land\varphi_2)&=ST^O_x(\varphi_1)\land ST^O _x(\varphi_2)\\
ST^O_x(\Diamond\varphi)&=\exists y(Rxy\land\neg(x\equiv v_0\land ST_y^{\lr{x,\bot}}(\psi_0))\land\\
&\bigwedge\limits_{0\le i\le n-1}\!\!\neg(x\equiv v_{i+1}\land ST_y^{\lr{v_0,\psi_0};...;\lr{v_i,\psi_i}}(\psi_{i+1}))\land ST_y^O(\varphi))\\
ST^O_x([-\varphi_1]\varphi_2)&=ST^{O;\lr{x,\varphi_1}}_x(\varphi_2)
\end{align*} 
\end{definition}

The key inductive clauses in Definition \ref{def-translation} concern $\Diamond$-formulas and $[-\;]$-formulas. Formula $\Diamond\varphi$ is translated as a first-order formula stating that the current point $x$ has a successor $y$ which is $\varphi$, and that this accessibility relation is not deleted by the operator $[-\;]$ indexed in the sequence $O$. The first-order translation for $[-\varphi_1]\varphi_2$ says that the translation of $\varphi_2$ is carried out with respect to the sequence $O;\lr{x,\varphi_1}$, and that this translation is realized at the current point $x$. 

According to Definition \ref{def-translation}, the index sequence $O$ may become longer and longer, but it is always finite. For each formula $\varphi$ of $\mathcal{L}_d$, $ST_x^{\lr{x,\bot}}(\varphi)$ yields a first-order formula with only $x$ free. Now we use an example to illustrate the translation.

\begin{example}\label{example-trans} Consider formula $\Diamond[-\Diamond p_1]\Box p_2$. Its translation runs as follows: 
\begin{align*}
ST_x^{\lr{x,\bot}}(\Diamond[-\Diamond p_1]\Box p_2)=& \exists y(Rxy\land\neg(x\equiv x\land ST_y^{\lr{x,\bot}}(\bot))\land\\ 
& ST_y^{\lr{x,\bot}}([-\Diamond p_1]\Box p_2)) \\
=&\exists y(Rxy\land\neg(x\equiv x\land ST_y^{\lr{x,\bot}}(\bot))\land\\ 
&ST_y^{\lr{x,\bot};\lr{y,\Diamond p_1}}(\Box p_2)) \\
=&\exists y(Rxy\land\neg(x\equiv x\land ST_y^{\lr{x,\bot}}(\bot))\land\\ 
&\forall z(Ryz\land\neg(y\equiv x\land ST_z^{\lr{x,\bot}}(\bot))\land \\
&\neg(y\equiv y\land ST_z^{\lr{x,\bot}}(\Diamond p_1) )\to ST_z^{\lr{x,\bot};\lr{y,\Diamond p_1}}(p_2))\\
=&\exists y(Rxy\land\neg(x\equiv x\land ST_y^{\lr{x,\bot}}(\bot))\land\\ 
&\forall z(Ryz\land\neg(y\equiv x\land ST_z^{\lr{x,\bot}}(\bot))\land \\
&\neg(y\equiv y\land\exists z'(Rzz'\land\neg(z\equiv x\land ST_{z'}^{\lr{x,\bot}}(\bot))\land\\&ST_{z'}^{\lr{x,\bot}}(p_1))\to ST_z^{\lr{x,\bot};\lr{y,\Diamond p_1}}(p_2))\\
=&\exists y(Rxy\land\neg(x\equiv x\land\neg y\equiv y)\land\forall z(Ryz\land\\ 
&\neg(y\equiv x\land\neg z\equiv z)\land \neg(y\equiv y\land\exists z'(Rzz'\land\\
&\neg(z\equiv x\land \neg z'\equiv z')\land P_1z')\to P_2z)
\end{align*}
\end{example}

The result is much complicated. Actually, it is equivalent to formula $\exists y(Rxy\land\forall z(Ryz\land\neg\exists z'(Rzz'\land P_1z')\to P_2z))$, which states that there exists a successor $y$ of the current point $x$ such that, for each successor $z$ of $y$, if $z$ does not has any $P_1$-successors, then $z$ is $P_2$. Example \ref{example-trans} can be considered as a small case illustrating that S$_d$ML is succinct notation for a complex part of first-order logic. In order to check the result, we now prove the correctness of Definition \ref{def-translation}.

\begin{theorem}[Correctness of the Standard Translation]\label{theorem-correctnessoftranslation}
Let $\lr{\mathcal{M},w}$ be a pointed model and $\varphi$ be a formula of $\mathcal{L}_d$, then
$$\mathcal{M},w\vDash \varphi\;\;{\textit{iff}}\;\; \mathcal{M}\vDash ST^{\lr{x,\bot}}_x(\varphi)[w].$$
\end{theorem}

\begin{proof}
The proof is by induction on the structure of $\varphi$. The cases for Boolean and modal connectives are straightforward. When $\varphi$ is $[-\varphi_1]\varphi_2$, the following equivalences hold:
\begin{align*}
\mathcal{M},w\vDash[-\varphi_1]\varphi_2&\;\;{\rm{iff}}\;\; \exists \mathcal{M}'\; {\rm{s.t.}} \lr{\mathcal{M},w}\xrightarrow[]{-\varphi_1}\lr{\mathcal{M}',w}\; {\rm{and}} \;\mathcal{M}',w\vDash\varphi_2\\
&\;\;{\rm{iff}}\;\;\exists \mathcal{M}'\; {\rm{s.t.}} \lr{\mathcal{M},w}\xrightarrow[]{-\varphi_1}\lr{\mathcal{M}',w}\; {\rm{and}} \;
\mathcal{M}'\vDash ST^{\lr{x,\bot}}_x(\varphi_2)[w]\\
&\;\;{\rm{iff}}\;\;\mathcal{M}\vDash ST_x^{\lr{x,\bot};\lr{x,\varphi_1}}(\varphi_2)[w]\\
&\;\;{\rm{iff}}\;\;\mathcal{M}\vDash ST^{\lr{x,\bot}}_x(\varphi)[w]
\end{align*}
The first equivalence follows from the semantics directly. By the inductive hypothesis, the second one holds. The last two equivalences hold by Definition \ref{def-translation}.
\end{proof}

\begin{remark}\label{remark-trans}
The first-order translation for S$_d$ML is quite different from that for SML. To translate a SML formula, it suffices to maintain a finite set of ordered pairs of nodes encoding the links already deleted (cf. \cite{sabotage}). However it fails for S$_d$ML, since the number of links cut by $[-\;]$ may be infinite. Besides, Example \ref{example-validity} shows that we should also take care of the order of $[-\;]$ in a formula. Our finite sequence of ordered pairs of nodes and properties solves these problems and yields a translation for S$_d$ML.
\end{remark}

Finally, we end by answering the question stated at the outset of this section, which follows directly from Definition \ref{def-translation} and Theorem \ref{theorem-correctnessoftranslation}:
\begin{corollary}\label{corollary}
By the completeness theorem for first-order logic, S$_d$ML is axiomatizable.
\end{corollary}

\section{Bisimulation and Expressivity for  S$_d$ML}\label{sec-bisimulationcharacterization}
Through the standard translation, we can translate a formula of S$_d$ML into first-order logic syntactically. In this section, we investigate the other aspect, i.e., model theories, for its expressive power. Let us begin with considering the notion of bisimulation for S$_d$ML.

\subsection{Bisimulation for S$_d$ML}\label{subsec-bisimulation}
After expanding the standard modal language $\mathcal{L}_\Box$ with the operator $[-\;]$, formulas of $\mathcal{L}_d$ are not invariant under the standard bisimulation any longer (cf. \cite{modallogic}). 

To show this, we first introduce a notion of \textit{definable sabotage modal equivalence} (notation, $\leftrightsquigarrow_d$) between pointed models: $\lr{\mathcal{M}_1,w}\leftrightsquigarrow_d\lr{\mathcal{M}_2,v}$ iff for each $\varphi\in\mathcal{L}_d$, $\mathcal{M}_1,w\vDash\varphi$ iff  $\mathcal{M}_2,v\vDash\varphi$.  

\begin{fact}\label{fact-bisimulation}
Formulas of $\mathcal{L}_d$ are not invariant under the standard bisimulation.
\end{fact}
\begin{proof}
It suffices to give an example. Consider two models $\mathcal{M}_1$ and $\mathcal{M}_2$ that are defined as depicted in the following figure:
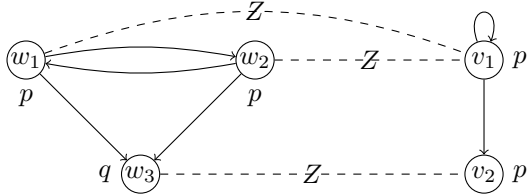
\begin{center}
\begin{tikzpicture} 
\node(a)[circle,draw,inner sep=0pt,minimum size=5mm,label=below:$p$] at (0,0) {$w_1$};
\node(b)[circle,draw,inner sep=0pt,minimum size=5mm,label=below:$p$] at (3,0){$w_2$};
\node(c)[circle,draw,inner sep=0pt,minimum size=5mm,label=left:$q$] at (1.5,-1.5){$w_3$};
\node(d)[circle,draw,inner sep=0pt,minimum size=5mm,label=right:$p$] at (6,0){$v_1$};
\node(e)[circle,draw,inner sep=0pt,minimum size=5mm,label=right:$p$] at (6,-1.5){$v_2$};
\draw[->](a) to[bend left=10] (b);
\draw[dashed](a) to [bend left=20]node {$Z$} (d);
\draw[->](b) to[bend left=10] (a);
\draw[dashed](b) to node {$Z$} (d);
\draw[dashed](c) to node{$Z$} (e);
\draw[->](a) to (c);
\draw[->](b) to (c);
\draw[->](d) to[in=70, out=110,looseness=8] (d);
\draw[->](d) to (e);
\end{tikzpicture} 
\end{center}
By the definition of standard bisimulation, we know that both  $\lr{\mathcal{M}_1,w_1}$ and $\lr{\mathcal{M}_1,w_2}$ are bisimilar to $\lr{\mathcal{M}_2,v_1}$, and that $\lr{\mathcal{M}_1,w_3}$ is bisimilar to $\lr{\mathcal{M}_2,v_2}$. However, we have $\mathcal{M}_1,w_1\vDash [-q] \Diamond\Diamond q$ and $\mathcal{M}_2,v_1\not\vDash [-q] \Diamond\Diamond q$. Therefore bisimulation does not imply definable sabotage modal equivalence.
\end{proof}

What is a suitable notion of bisimulation for S$_d$ML? Now we introduce a new notion of \textit{definable sabotage bisimulation} (\textit{d-bisimulation}). Here is the formal definition.

\begin{definition}[d-bisimulation]\label{def-bisimulation}
Let $\mathcal{M}_1=\lr{W_1,R_1,V_1}$ and $\mathcal{M}_2=\lr{W_2,R_2,V_2}$ be two models. A non-empty relation $Z_d$ is a \textit{d-bisimulation} between pointed models $\lr{\mathcal{M}_1,w}$ and $\lr{\mathcal{M}_2,v}$ (notation, $Z_d:\lr{\mathcal{M}_1,w}\underline{\leftrightarrow}_d\lr{\mathcal{M}_2,v}$) if the following five conditions are satisfied:

{\rm{\textbf{Atom}}}: If $\lr{\mathcal{M}_1,w}Z_d\lr{\mathcal{M}_2,v}$, then $\mathcal{M}_1,w\vDash p$ iff $\mathcal{M}_2,v\vDash p$, for each $p\in\textbf{\rm{\textbf{P}}}$.

{\rm{\textbf{Zig}$_{\Diamond}$}}: If $\lr{\mathcal{M}_1,w}Z_d\lr{\mathcal{M}_2,v}$ and there exists $w'\in W_1$ such that $R_1ww'$, then there exists $v'\in W_2$ such that $R_2vv'$ and $\lr{\mathcal{M}_1,w'}Z_d\lr{\mathcal{M}_2,v'}$.

{\rm{\textbf{Zag}$_{\Diamond}$}}: If $\lr{\mathcal{M}_1,w}Z_d\lr{\mathcal{M}_2,v}$ and there exists $v'\in W_2$ such that $R_2vv'$, then there exists $w'\in W_1$ such that  $R_1ww'$ and $\lr{\mathcal{M}_1,w'}Z_d\lr{\mathcal{M}_2,v'}$.

{\rm{\textbf{Zig}$_{[-\;]}$}}: For each $\varphi\in\mathcal{L}_d$, if $\lr{\mathcal{M}_1,w}Z_d\lr{\mathcal{M}_2,v}$ and there exists $\mathcal{M}_1'$ such that  $\lr{\mathcal{M}_1,w}\xrightarrow[]{-\varphi}\lr{\mathcal{M}_1',w}$, then there exists $\mathcal{M}_2'$ such that $\lr{\mathcal{M}_2,v}\xrightarrow[]{-\varphi}\lr{\mathcal{M}_2',v}$ and $\lr{\mathcal{M}_1',w}Z_d\lr{\mathcal{M}_2',v}$.

{\rm{\textbf{Zag}$_{[-\;]}$}}: For each $\varphi\in\mathcal{L}_d$, if $\lr{\mathcal{M}_1,w}Z_d\lr{\mathcal{M}_2,v}$ and there exists $\mathcal{M}_2'$ such that  $\lr{\mathcal{M}_2,v}\xrightarrow[]{-\varphi}\lr{\mathcal{M}_2',v}$, then there exists $\mathcal{M}_1'$ such that $\lr{\mathcal{M}_1,w}\xrightarrow[]{-\varphi}\lr{\mathcal{M}_1',w}$ and $\lr{\mathcal{M}_1',w}Z_d\lr{\mathcal{M}_2',v}$.

For brevity, we write $\lr{\mathcal{M}_1,w}\underline{\leftrightarrow}_d\lr{\mathcal{M}_2,v}$ if there exists a d-bisimulation $Z_d$ such that $\lr{\mathcal{M}_1,w}Z_d\lr{\mathcal{M}_2,v}$.
\end{definition}

Here the conditions for $\Diamond$ are as usual, and they do not change the model but change the evaluation point along the accessibility relation. While, the conditions for $[-\;]$ keep the evaluation point fixed but remove some links from the model. In the standard modal logic, given any two models $\mathcal{M}$ and $\mathcal{N}$, there always exists a bisimulation called \textit{largest bisimulation}, i.e., the set-theoretic union of all bisimulation relations between $\mathcal{M}$ and $\mathcal{N}$ (see \cite{openmind}). By Definition \ref{def-bisimulation}, it is not hard to see that this also holds for the new notion: for any two models, there is a largest d-bisimulation between them. This result is useful in various aspects, say, it can help us to simplify given models to smaller equivalent ones. 

As a concrete illustration of the notion introduced here, it is easy to see that the pointed models $\lr{\mathcal{M}_1,w_1}$ and $\lr{\mathcal{M}_2,v_1}$ in the proof of Fact \ref{fact-bisimulation} are not d-bisimilar. 

Next we show that formulas of S$_d$ML are invariant for d-bisimulation:
 
 \begin{theorem}[$\underline{\leftrightarrow}_d\subseteq\leftrightsquigarrow_d$]\label{theorem-bisimtoequiv}
For any $\lr{\mathcal{M}_1,w}$ and $\lr{\mathcal{M}_2,v}$, if $\lr{\mathcal{M}_1,w}\underline{\leftrightarrow}_d\lr{\mathcal{M}_2,v}$, then $\lr{\mathcal{M}_1,w}\leftrightsquigarrow_d\lr{\mathcal{M}_2,v}$. 
\end{theorem}

\begin{proof}
We prove it by induction on the syntax of $\varphi$. Let $\lr{\mathcal{M}_1,w}\underline{\leftrightarrow}_d\lr{\mathcal{M}_2,v}$.

(1). $\varphi\in\textbf{P}$. By Definition \ref{def-bisimulation}, it holds directly that $\mathcal{M}_1,w\vDash\varphi$ iff $\mathcal{M}_2,v\vDash\varphi$.

(2). $\varphi$ is $\neg\psi$. By the inductive hypothesis, $\mathcal{M}_1,w\vDash\psi$ iff $\mathcal{M}_2,v\vDash\psi$. Consequently, we know that $\mathcal{M}_1,w\vDash\varphi$ iff $\mathcal{M}_2,v\vDash\varphi$.

(3). $\varphi$ is $\varphi_1\land\varphi_2$. By the inductive hypothesis, for each $i\in\{1,2\}$,  $\mathcal{M}_1,w\vDash\varphi_i$ iff $\mathcal{M}_2,v\vDash\varphi_i$. Thus it holds that $\mathcal{M}_1,w\vDash\varphi$ iff $\mathcal{M}_2,v\vDash\varphi$.

(4). $\varphi$ is $\Diamond\psi$. If $\mathcal{M}_1,w\vDash\varphi$, then there exists $w_1\in W_1$ such that $R_1ww_1$ and  $\mathcal{M}_1,w_1\vDash\psi$. By \textbf{Zig}$_\Diamond$, there exists $v_1\in W_2$ s.t. $R_2vv_1$ and $\lr{\mathcal{M}_1,w_1}\underline{\leftrightarrow}_d\lr{\mathcal{M}_2,v_1}$. By the inductive hypothesis, $\mathcal{M}_1,w_1\vDash\psi$ iff $\mathcal{M}_2,v_1\vDash\psi$. It is followed by $\mathcal{M}_2,v_1\vDash\psi$ immediately. Consequently it holds that $\mathcal{M}_2,v\vDash\varphi$. Similarly, we can obtain $\mathcal{M}_1,w\vDash\varphi$ from $\mathcal{M}_2,v\vDash\varphi$ by \textbf{Zag}$_\Diamond$. 

(5). $\varphi$ is $[-\varphi_1]\varphi_2$. If $\mathcal{M}_1,w\vDash\varphi$, then there is a $\mathcal{M}_1'$ s.t. $\lr{\mathcal{M}_1,w}\xrightarrow[]{-\varphi_1}\lr{\mathcal{M}_1',w}$ and $\mathcal{M}_1',w \vDash\varphi_2$. By \textbf{Zig}$_{[-\;]}$, there is some $\mathcal{M}_2'$ such that  $\lr{\mathcal{M}_2,v}\xrightarrow[]{-\varphi_1}\lr{\mathcal{M}_2',v}$ and $\lr{\mathcal{M}_1',w}\underline{\leftrightarrow}_d\lr{\mathcal{M}_2',v}$. By the inductive hypothesis, $\mathcal{M}_1',w\vDash\varphi_2$ iff $\mathcal{M}_2',v\vDash\varphi_2$. Hence it holds that $\mathcal{M}_2,v\vDash\varphi$. Similarly, by \textbf{Zag}$_{[-\;]}$, $\mathcal{M}_1,w\vDash\varphi$ follows from $\mathcal{M}_2,v\vDash\varphi$.
\end{proof}

As an application of Theorem \ref{theorem-bisimtoequiv}, let us consider a simple example:

\begin{example}\label{example-bisimulation}
Consider two models $\mathcal{M}_1$ and $\mathcal{M}_2$ defined respectively as follows:
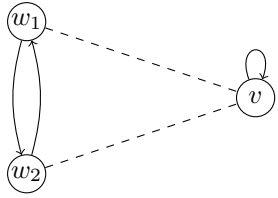
\begin{center}
\begin{tikzpicture}[every node/.style={circle,draw,inner sep=0pt,minimum size=5mm}]
\node(a) at (0,2) {$w_1$};
\node(b) at (0,0){$w_2$};
\node(c) at (3,1){$v$};
\draw[->](a) to[bend right=15] (b);
\draw[->](b) to[bend right=15] (a);
\draw[->](c) to[in=70, out=110,looseness=8] (c);
\draw[dashed](a) to (c);
\draw[dashed](b) to (c);
\end{tikzpicture}    
\end{center}
By Definition \ref{def-bisimulation}, it holds that $\lr{\mathcal{M}_1,w_1}\underline{\leftrightarrow}_d\lr{\mathcal{M}_2,v}$ and $\lr{\mathcal{M}_1,w_2}\underline{\leftrightarrow}_d\lr{\mathcal{M}_2,v}$ (the d-bisimulation runs via the dashed lines). From Theorem \ref{theorem-bisimtoequiv}, we know that $\lr{\mathcal{M}_1,w_1}\leftrightsquigarrow_d \lr{\mathcal{M}_2,v}$ and $\lr{\mathcal{M}_1,w_2}\leftrightsquigarrow_d\lr{\mathcal{M}_2,v}$. Therefore, S$_d$ML cannot distinguish between nodes $w_{1(2)}$ and $v$.
\end{example}

Furthermore, for \textit{$\omega$-saturated} models, the converse of Theorem \ref{theorem-bisimtoequiv} holds as well. For each finite set $Y$, we denote the expansion of $\mathcal{L}_{1}$ with a set $Y$ of constants with $\mathcal{L}_1^Y$, and denote the expansion of $\mathcal{M}$ to $\mathcal{L}_1^Y$ with $\mathcal{M}^Y$.

\begin{definition}[$\omega$-saturation]\label{def-omega}
A model $\mathcal{M}=\lr{W,R,V}$ is \textit{$\omega$-saturated} if, for every finite subset $Y$ of $W$, the expansion $\mathcal{M}^Y$ realizes every set $\Gamma(x)$ of $\mathcal{L}_1^Y$-formulas whose finite subsets $\Gamma'(x)$ are all realized in $\mathcal{M}^Y$.
\end{definition}

Not all models are $\omega$-saturated, but every model can be extended to an $\omega$-saturated model with the same first-order theory (see \cite{modeltheory}). From Definition \ref{def-translation}, we know that each model $\mathcal{M}$ has an $\omega$-saturated extension with the same theory of S$_d$ML. For brevity, we use the set $\mathbb{T}^d(\mathcal{M},w)=\{\varphi\in\mathcal{L}_d\mid\mathcal{M},w\vDash\varphi\}$ of $\mathcal{L}_{d}$-formulas to denote the theory of $w$ in $\mathcal{M}$. By Definition \ref{def-omega}, we have the following result.

\begin{theorem}[$\leftrightsquigarrow_d\subseteq\underline{\leftrightarrow}_d$] \label{theorem-equivtobisim} For any two $\omega$-saturated pointed models $\lr{\mathcal{M}_1,w}$ and $\lr{\mathcal{M}_2,v}$, if $\lr{\mathcal{M}_1,w}\leftrightsquigarrow_d\lr{\mathcal{M}_2,v}$, then  $\lr{\mathcal{M}_1,w}\underline{\leftrightarrow}_d\lr{\mathcal{M}_2,v}$.
\end{theorem}

\begin{proof}
We prove this by showing that $\leftrightsquigarrow_d$ satisfies the definition of d-bisimulation.

(1). For each $p\in\textbf{P}$, by the definition of $\leftrightsquigarrow_d$, it holds that $\mathcal{M}_1,w\vDash p$ iff $\mathcal{M}_2,v\vDash p$. This satisfies the condition of \textbf{Atom}. 

(2). Let $w_1\in W_1$ such that $R_1ww_1$. We show that point $v$ has a successor $v_1$ with $\lr{\mathcal{M}_1,w_1}\leftrightsquigarrow_d\lr{\mathcal{M}_2,v_1}$. For each finite subset $\Gamma$ of $\mathbb{T}^d(\mathcal{M}_1,w_1)$, it holds that:
\begin{align*}
\mathcal{M}_1,w\vDash\Diamond\bigwedge\Gamma &\;\;{\rm{iff}}\;\;  \mathcal{M}_2,v\vDash\Diamond\bigwedge\Gamma\\
&\;\;{\rm{iff}}\;\;\mathcal{M}_2\vDash ST_x^{\lr{x,\bot}}(\Diamond\bigwedge\Gamma)[v]\\
&\;\;{\rm{iff}}\;\;\mathcal{M}_2\vDash\exists y(Rxy\land ST_y^{\lr{x,\bot}}(\bigwedge\Gamma))[v]   
\end{align*}
Therefore every finite subset $\Gamma$ of $\mathbb{T}^d(\mathcal{M}_1,w_1)$ is satisfiable in the set of successors of node $v$. From Definition \ref{def-omega}, we know that $v$ has a successor $v_1$ where $\mathbb{T}^d(\mathcal{M}_1,w_1)$ is true. Thus, $\lr{\mathcal{M}_1,w_1}\leftrightsquigarrow_d\lr{\mathcal{M}_2,v_1}$. The proof of the \textbf{Zig$_\Diamond$} clause is completed.

(3). Similar to (2), we can prove that the condition of
\textbf{Zag$_\Diamond$} is satisfied.

(4). Let $\lr{\mathcal{M}_1',w}$ be a pointed model and $\varphi\in\mathcal{L}_d$ such that $\lr{\mathcal{M}_1,w}\xrightarrow[]{-\varphi}\lr{\mathcal{M}_1',w}$. We prove the \textbf{Zig$_{[-\;]}$} clause by showing there exists $\mathcal{M}_2'$ with  $\lr{\mathcal{M}_2,v}\xrightarrow[]{-\varphi}\lr{\mathcal{M}_2',v}$ and $\lr{\mathcal{M}_1',w}\leftrightsquigarrow_d\lr{\mathcal{M}_2',v}$. For each finite subset $\Gamma$ of $\mathbb{T}^d(\mathcal{M}_1',w)$, the following sequence of equivalences holds:
\begin{align*}
\mathcal{M}_1,w\vDash [-\varphi]\bigwedge\Gamma &\;\;{\rm{iff}}\;\;\mathcal{M}_2,v\vDash[-\varphi]\bigwedge\Gamma\\
&\;\;{\rm{iff}}\;\;\mathcal{M}_2\vDash ST_x^{\lr{x,\bot}}([-\varphi]\bigwedge\Gamma)[v]\\
&\;\;{\rm{iff}}\;\;\mathcal{M}_2\vDash ST_x^{\lr{x,\bot};\lr{x,\varphi}}(\bigwedge\Gamma)[v]   
\end{align*}
Hence each finite subset of  $\mathbb{T}^d(\mathcal{M}_1',w)$ is true at $\lr{\mathcal{M}_2',v}$, where $\lr{\mathcal{M}_2,v}\xrightarrow[]{-\varphi}\lr{\mathcal{M}_2',v}$. By Definition \ref{def-omega}, $\mathbb{T}^d(\mathcal{M}_1',w)$ is true at $\lr{\mathcal{M}_2',v}$. It is followed by  $\lr{\mathcal{M}_1',w}\leftrightsquigarrow_d\lr{\mathcal{M}_2',v}$.  

(5). Similar to (4), we can show that the condition of
\textbf{Zag$_{[-\;]}$} is satisfied.

Thus, we conclude that $\lr{\mathcal{M}_1,w}\underline{\leftrightarrow}_d\lr{\mathcal{M}_2,v}$. The proof is completed.
\end{proof}

\subsection{Characterization of S$_d$ML}\label{subsec-characterization}

By the notion of d-bisimulation, we can characterize S$_d$ML as the one-free-variable fragment of FOL that is invariant for d-bisimulation, where a first-order formula $\alpha(x)$ is invariant for d-bisimulation means that for all pointed models $\lr{\mathcal{M}_1,w_1}$ and $\lr{\mathcal{M}_2,w_2}$ such that $\lr{\mathcal{M}_1,w_1}\underline{\leftrightarrow}_d\lr{\mathcal{M}_2,w_2}$, it holds that $\mathcal{M}_1\vDash\alpha(x)[w_1]$ iff $\mathcal{M}_2\vDash\alpha(x)[w_2]$. 

\begin{theorem}[Characterization of S$_d$ML by d-bisimulation Invariance] \label{theorem-characterization} An $\mathcal{L}_1$-formula is equivalent to the translation of an $\mathcal{L}_d$-formula iff it is invariant for d-bisimulation.
\end{theorem}

\begin{proof}
The direction from left to right holds directly by Theorem \ref{theorem-bisimtoequiv}. For the converse direction, let $\alpha$ be an $\mathcal{L}_1$-formula with one free variable $x$. Assume that $\alpha$ is invariant for d-bisimulation. Now we consider the following set:
$$\mathbb{C}_d(\alpha)=\{ST_x^{\lr{x,\bot}}(\varphi)|\varphi\in\mathcal{L}_d\ {\rm and}\ \alpha\vDash ST_x^{\lr{x,\bot}}(\varphi)\}.$$
The result holds from the following two claims:
\begin{itemize}
\item[(i).] If $\mathbb{C}_d(\alpha)\vDash\alpha$, then $\alpha$ is equivalent to the translation of an $\mathcal{L}_d$-formula.
\item[(ii).] $\mathbb{C}_d(\alpha)\vDash\alpha$, i.e., for any pointed model $\lr{\mathcal{M},w}$, $\mathcal{M}\vDash\mathbb{C}_d(\alpha)[w]$ entails $\mathcal{M}\vDash\alpha[w]$.
\end{itemize}

We show (i) first. Suppose that $\mathbb{C}_d(\alpha)\vDash\alpha$. From the compactness and deduction theorems of first-order logic, it holds that $\vDash\bigwedge\Gamma\to\alpha$ for some finite  subset $\Gamma$ of $\mathbb{C}_d(\alpha)$. The converse can be shown by the definition of $\mathbb{C}_d(\alpha)$: $\vDash \alpha\to\bigwedge\Gamma$. Thus it holds that $\vDash\alpha\leftrightarrow\bigwedge\Gamma$ proving the claim.

As to the claim (ii), let $\lr{\mathcal{M},w}$ be a pointed model such that $\mathcal{M}\vDash\mathbb{C}_d(\alpha)[w]$. Consider the set $\Sigma=ST_x^{\lr{x,\bot}}(\mathbb{T}^d(\mathcal{M},w))\cup\{\alpha\}$. We now show that:
\begin{itemize}
\item[(a).] The set $\Sigma$ is consistent.
\item[(b).] $\mathcal{M}\vDash\alpha[w]$, thus proving claim (ii).
\end{itemize}

Suppose that $\Sigma$ is not consistent. By the compactness of first-order logic, it follows that $\vDash\alpha\to\neg\bigwedge\Gamma$ for some finite subset $\Gamma$ of $\Sigma$. But then, by the definition of $\mathbb{C}_d(\alpha)$, we obtain $\neg\bigwedge\Gamma\in\mathbb{C}_d(\alpha)$, which is followed by $\neg\bigwedge\Gamma\in ST_x^{\lr{x,\bot}}(\mathbb{T}^d(\mathcal{M},w))$. However, it contradicts to $\Gamma\subseteq ST_x^{\lr{x,\bot}}(\mathbb{T}^d(\mathcal{M},w))$. Hence (a) holds.

Now we show that (b) holds as well. Since $\Sigma$ is consistent, it can be realized by some pointed model, say, $\lr{\mathcal{M}',w'}$. Note that both the pointed models have same theories, thus $\lr{\mathcal{M},w}\leftrightsquigarrow_d\lr{\mathcal{M}',w'}$. Now take two $\omega$-saturated elementary extensions $\lr{\mathcal{M}_{\omega},w}$ and $\lr{\mathcal{M}'_{\omega},w'}$ of $\lr{\mathcal{M},w}$ and $\lr{\mathcal{M}',w'}$ respectively. It can be shown that such extensions always exist (see \cite{modeltheory}). By the invariance of first-order logic under elementary extensions, from $\mathcal{M}'\vDash\alpha[w']$ we know $\mathcal{M}'_{\omega}\vDash\alpha[w']$. Moreover, by Theorem \ref{theorem-equivtobisim} and the assumption that $\alpha$ is invariant for d-bisimulation, we have $\mathcal{M}_{\omega}\vDash\alpha[w]$. By the elementary extension, we obtain $\mathcal{M}\vDash\alpha[w]$ that entails the claim (ii). Consequently, the proof is completed.
\end{proof}

Just as with SML, the key model-theoretic argument using saturation needed special care, but now with new modifications matching the above translation of S$_d$ML (cf. \cite{sabotage}).

\subsection{Exploring Expressive Power}\label{subsec-compare}
So far, we have already been able to show whether or not a first-order property belongs to the fragment identified by Theorem \ref{theorem-characterization}. In this section, we show several concrete examples, which will also present a comparison between S$_d$ML and SML with respect to their expressive power on models.

\begin{example}\label{example-seeback}
Consider the first-order property $\alpha_1(x)$ `The current point is irreflexive and not a dead end. Each of its successors only has access to it', i.e., $\alpha_1(x):=\neg Rxx\land\exists yRxy\land\forall y(Rxy\to Ryx\land\forall z(Ryz\to z\equiv x))$. From Example \ref{example-bisimulation}, we know that this property is not invariant for d-bisimulation. For instance, formula $\alpha_1(x)$ is true at state $w_1$ in $\mathcal{M}_1$ but fails at $v$ in $\mathcal{M}_2$. Thus this property is not definable in S$_d$ML.
\end{example}

Interestingly, the result may be quite different if we change the first-order property in Example \ref{example-seeback} slightly, say,

\begin{fact}\label{fact-seebackwithothersuccessors}
The first-order property $\alpha_1^+(x)$ `The current point is irreflexive and not a dead end. Some of its successors are dead ends, the others only have access to dead ends and the current point', i.e., $\alpha_1^+(x):=\neg Rxx\land\exists y(Rxy\land\neg\exists zRyz)\land\exists y(Rxy\land \exists zRyz)\land\forall y(Rxy\to\neg\exists zRyz\lor(Ryx\land\exists z(Ryz\land\neg\exists uRzu)\land\forall z(Ryz\to z\equiv x\lor\neg\exists uRzu)))$, is definable in S$_d$ML.
\end{fact}

\begin{proof}
Consider the following formulas of S$_d$ML:
\begin{align*}
 (B_1)\qquad &\Diamond\Box\bot\land\Diamond\Diamond\top\\
 (B_2)\qquad & \Box(\Diamond\top\to\Diamond\Box\bot\land\Diamond(\Diamond\Box\bot\land\Diamond\Diamond\top)\land\Box(\Box\bot\lor(\Diamond\Box\bot\land\Diamond\Diamond\top)))\\
 (B_3)\qquad & [-\Box \bot]\Box(\Diamond\Box\bot\land\Box(\neg\Box\bot\to\neg\Diamond\Box\bot))
\end{align*}
Let $\varphi_1^+:=(B_1\land B_2\land B_3)$. This formula is satisfiable, say, it is true at $\lr{\mathcal{M}_1,w_1}$ in the proof of Fact \ref{fact-bisimulation}. Let $\lr{\mathcal{M},u}$ be a pointed model. It is not hard to see that $\mathcal{M},u\vDash \varphi_1^+$ if $\mathcal{M}\vDash\alpha_1^+(x)[u]$. Now assume that $\mathcal{M},u\vDash \varphi_1^+$. Formula $(B_1)$ states that, the current point $u$ has some successors $u_1$ that are dead ends, and some successors $u_2$ which have successors. By $(B_2)$, each $u_2$ reaches some dead end $u_3$, and some point $u_4$ which is similar to $u$: it has some successors which are dead ends, and some successors that also have successors. After cutting the links from node $u$ to the dead ends, from $(B_3)$ it holds that $u_2$ still can see some dead ends, and that $u_4$ cannot reach dead ends any longer. Therefore we obtain $u_2\not=u$ and $u_4=u$, consequently, $\mathcal{M}\vDash\alpha_1^+(x)[u]$. So we conclude that  $\mathcal{M}\vDash\alpha_1^+(x)[u]$ iff $\mathcal{M},u\vDash\varphi_1^+$ for any pointed model $\lr{\mathcal{M},u}$. 
\end{proof}

Through observation, we can find that the property $\alpha_1^+(x)$ expands the current point and its successors in $\alpha_1(x)$ with some successors that are dead ends. But the former one is definable in S$_d$ML and the latter one is not. What is the reason for this? 

Suppose that $\lr{\mathcal{M},u}$ be a pointed model that is d-bisimilar to $\lr{\mathcal{M}_1,w_1}$ in the proof of Fact \ref{fact-bisimulation}. By Definition \ref{def-bisimulation}, we know that $u$ can reach some dead end $u_1$, and some $u_2$ that has access to some dead ends. Except those dead ends, $u_2$ can also see some point $u_3$ that is similar to $u$: $u_3$ can reach some dead end and some node that has successors. Further more, after cutting the links from $u$ to the dead ends, $u_2$ still can see some dead ends, but $u_3$ cannot reach any dead ends now. So we have $u_2\not=u$ and $u_3=u$. In such a way, we conclude that the property $\alpha_1^+(x)$ is invariant under d-bisimulation. 

\begin{example}\label{example-count}
Consider the FOL property `There exist $n$ successors of the current point'. This property is not invariant for d-bisimulation. For instance, in the following models:
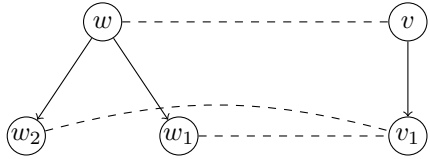
\begin{center}
\begin{tikzpicture}[every node/.style={circle,draw,inner sep=0pt,minimum size=5mm}]
\node(a) at (0,0) {$w$};
\node(b) at (-1,-1.5){$w_2$};
\node(c) at (1,-1.5){$w_1$};
\node(d) at (4,0){$v$};
\node(e) at (4,-1.5){$v_1$};
\draw[->](a) to (b);
\draw[->](a) to (c);
\draw[->](d) to (e);
\draw[dashed](a) to (d);
\draw[dashed](b) to [bend left=15](e);
\draw[dashed](c) to (e);
\end{tikzpicture} 
\end{center}
the property `there exist $2$ successors' is true at point $w$ in the model to the left, but it fails at $v$ to the right. Hence it is not definable in S$_d$ML.
\end{example}

In contrast, as noted in \cite{sabotage}, SML can count successors of the current state,
and it can also define the length of a cycle. That is, for each positive natural number $n$, there exists a SML formula $\varphi$ such that, for any $\mathcal{M}=\lr{W,R,V}$ and $w\in W$, $\mathcal{M},w\vDash\varphi$ iff $\lr{W,R}$ is a cycle of length $n$. Is this property definable in S$_d$ML?

\begin{example}\label{example-cycle}
Recall the two models displayed in Example \ref{example-bisimulation}. The underlying frame of $\mathcal{M}_1$ is a cycle of length 2, while that of $\mathcal{M}_2$ is a cycle of length 1. So S$_d$ML cannot define the length of a cycle.
\end{example}
 
 Intuitively, these differences between S$_d$ML and SML stem from the features of $[-\;]$ and the standard sabotage modality $\blacklozenge$. In SML, each occurrence of $\blacklozenge$ in a formula deletes exactly one link. While, in S$_d$ML, $[-\;]$ operates uniformly, which blocks the logic to define the first-order properties in Example \ref{example-count}-\ref{example-cycle}. However, this does not mean that S$_d$ML is less expressive than SML with respect to models. Actually the notion of bisimulation for sabotage modal logic is not an extension of d-bisimulation. When tackling with cases involving infinite, operator $[-\;]$ may show more strength. Here is an example.
 
 \begin{example}\label{example-infinite}
 We establish a model in the following way: 
 
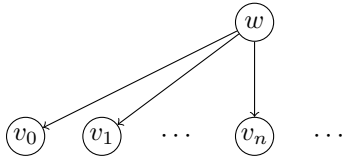
\begin{center}
  \begin{tikzpicture} 
\node(a)[circle,draw,inner sep=0pt,minimum size=5mm] at (0,0) {$w$};
\node(b)[circle,draw,inner sep=0pt,minimum size=5mm] at (-3,-1.5) {$v_0$};
\node(c)[circle,draw,inner sep=0pt,minimum size=5mm] at (-2,-1.5) {$v_1$};
\node(d)[circle,inner sep=0pt,minimum size=5mm] at (-1,-1.5) {$\cdots$};
\node(d)[circle,draw,inner sep=0pt,minimum size=5mm] at (0,-1.5) {$v_n$};
\node(e)[circle,inner sep=0pt,minimum size=5mm] at (1,-1.5) {$\cdots$};
\draw[->](a) to (b);
\draw[->](a) to (c);
\draw[->](a) to (d);
\end{tikzpicture}
 \end{center}
 As we can see, node $w$ has countable successors. By the truth condition for $[-\;]$, formula $[-\top]\Box\bot$ is true at node $w$, which says that all links starting at $w$ are cut by the operator $[-\;]$. However there is no such a formula $\varphi$ of SML that can do this: each occurrence of $\blacklozenge$ cuts one link, but the number of $\blacklozenge$ occurring in a formula is always finite.
 \end{example} 
 
From Examples \ref{example-count}-\ref{example-infinite}, we then get the following conclusion.  
 
 \begin{fact}\label{fact-comparisonwithSML}
S$_d$ML and SML are not comparable in their expressive power on models. 
 \end{fact}

\section{From S$_d$ML to Hybrid Logics}\label{sec-hybrid}
While an effective first-order translation shows that validity in S$_d$ML is effectively axiomatizable, it gives no concrete information about a more `modal' complete set of proof principles. In this section, following the techniques developed by some dynamic-epistemic logics (cf. e.g. \cite{baltag}), we try to axiomatize S$_d$ML by means of recursion axioms. 

The principles for Boolean cases are as usual. However, as for $[-\varphi]\Box\psi$, there is a problem. From the typical method of recursion axioms used in dynamic-epistemic logic, we know that dynamic operators can be pushed inside through standard modalities. But it fails for S$_d$ML, since that after pushing $[-\;]$ under a standard modality over successors of the current world, the model change is not local in the successors any longer and it takes place somewhere else (cf. \cite{sabotage}).
 
 Hence the principle for $[-\varphi]\Box\psi$ should illustrate the position where the change happens. To do so, a natural method is to seek help from hybrid logics, which enable us to name nodes in a model. Consider the hybrid logic with \textit{nominals}, \textit{at operator} $@$ and \textit{down-arrow operator} $\downarrow$, which is denoted by $\mathcal{H}(\downarrow)$. With its formulas of the form $\downarrow x\Box\downarrow y\varphi$, we can manipulate links by naming pairs of points (see \cite{Relation change}).

\subsection{S$_d$ML and Hybrid Logics}\label{subsec-hybridtranslation}

As a warm-up, we briefly discuss the relation between S$_d$ML and hybrid logics. In particular,  the following translation illustrates that S$_d$ML can be reduced to $\mathcal{H}(\downarrow)$. Similar to the standard translation, a finite sequence $O$ will be used.

\begin{definition}[The Hybrid Translation for S$_d$ML]\label{def-hybridtranslation} Let $O$ be a finite sequence of pairs of variables of nominals and properties, denoted with $\lr{x_0,\psi_0};...;\lr{x_i,\psi_i};...;\lr{x_n,\psi_n}(0\le i\le n)$. The translation $T^O:\mathcal{L}_d\to \mathcal{H}(\downarrow)$ is recursively  defined in the following way: 
\begin{eqnarray} 
T^O(p)&=&p\nonumber\\
T^O(\top)&=&\top\nonumber\\
T^O(\neg\varphi)&=&\neg T^O(\varphi)\nonumber\\
T^O(\varphi_1\land\varphi_2)&=&T^O(\varphi_1)\land T^O(\varphi_2)\nonumber\\
T^O(\Diamond\varphi)&=&\downarrow x\Diamond(\neg(@_xx_0\land T^{\lr{x_0,\bot}}(\psi_0))\land\nonumber\\
&&\bigwedge\limits_{0\leqslant i\leqslant n-1}\!\!\neg(@_xx_{i+1}\land T^{\lr{x_0,\psi_0};...;\lr{x_i,\psi_i}}(\psi_{i+1}))\land T^O(\varphi))\nonumber\\
T^O([-\psi]\varphi)&=&\downarrow x T^{O;\lr{x,\psi}}(\varphi)\nonumber
\end{eqnarray}
\end{definition}

In fact, the truth value of a $\mathcal{H}(\downarrow)$-formula in some model may depend on the valuation of nominals occurring in it. However, this is not problematic: by Definition \ref{def-hybridtranslation}, for each $\varphi\in\mathcal{L}_d$, $T^{\lr{x,\bot}}(\varphi)$ yields a $\mathcal{H}(\downarrow)$-formula with no free variables of nominals. For brevity, we will leave out the assignment of values to variables in models of $\mathcal{H}(\downarrow)$ if there is no ambiguity.  Now we show the correctness of Definition \ref{def-hybridtranslation}.

\begin{theorem}[Correctness of the Hybrid  Translation]\label{theorem-correctnessofhybridtranslation}
Let $\lr{\mathcal{M},w}$ be a pointed model and $\varphi$ be a formula of $\mathcal{L}_d$, then
$$\mathcal{M},w\vDash\varphi\;\;{\textit{iff}}\;\;\mathcal{M},w\vDash T^{\lr{x,\bot}}(\varphi).$$
\end{theorem}

\begin{proof}
The proof is by induction on the structure of $\varphi$. The Boolean cases are straightforward, and we only show the non-trivial cases.

(1). When $\varphi$ is $\Diamond\psi$, the following equivalences hold:
\begin{align*}
\mathcal{M},w\vDash\varphi &\;\;{\rm{iff}}\;\; {\rm{there\; exists}}\; v\in W \;{\rm{s.t.}}\; Rwv \;{\rm{and}}\;\mathcal{M},v\vDash\psi\\ 
&\;\;{\rm{iff}}\;\; {\rm{there\; exists}}\; v\in W\;{\rm{s.t.}}\; Rwv\;{\rm{and}}\;\mathcal{M},v\vDash T^{\lr{x,\bot}}(\psi)\\
&\;\;{\rm{iff}}\;\; \mathcal{M},w\vDash\Diamond T^{\lr{x,\bot}}(\psi)\\
&\;\;{\rm{iff}}\;\;  \mathcal{M},w\vDash\downarrow x\Diamond(\neg(@_xx\land T^{\lr{x,\bot}}(\bot))\land T^{\lr{x,\bot}}(\psi))\\
&\;\;{\rm{iff}}\;\;  \mathcal{M},w\vDash T^{\lr{x,\bot}}(\varphi)
\end{align*}
The first equivalence holds by the semantics of $\mathcal{L}_d$. The second one follows from the inductive hypothesis. The third and fourth equivalences follow by the semantics of $\mathcal{H}(\downarrow)$. The last one holds by Definition \ref{def-hybridtranslation}. 

(2). When $\varphi$ is $[-\varphi_1]\varphi_2$, we have the following equivalences:
\begin{align*}
\mathcal{M},w\vDash[-\varphi_1]\varphi_2&\;\;{\rm{iff}}\;\; \exists\mathcal{M}'\; {\rm{s.t.}} \lr{\mathcal{M},w}\xrightarrow[]{-\varphi_1}\lr{\mathcal{M}',w}\; {\rm{and}} \;\mathcal{M}',w\vDash\varphi_2\\
&\;\;{\rm{iff}}\;\; \exists \mathcal{M}'\; {\rm{s.t.}} \lr{\mathcal{M},w}\xrightarrow[]{-\varphi_1}\lr{\mathcal{M}',w}\; {\rm{and}} \;\lr{\mathcal{M}',w}\vDash T^{\lr{x,\bot}}(\varphi_2)\\
&\;\;{\rm{iff}}\;\;\lr{\mathcal{M},w}\vDash T^{\lr{x,\bot};\lr{x,\varphi_1}}(\varphi_2)\\
&\;\;{\rm{iff}}\;\;\lr{\mathcal{M},w}\vDash T^{\lr{x,\bot}}(\varphi)    
\end{align*}
The first equivalence follows directly from the semantics of $\mathcal{L}_d$. The second one holds by the inductive hypothesis. The last two equivalences follow by Definition \ref{def-hybridtranslation}.

Therefore, for each $\varphi\in\mathcal{L}_d$, it holds that $\mathcal{M},w\vDash\varphi$ iff $\mathcal{M},w\vDash T^{\lr{x,\bot}}(\varphi)$.
\end{proof}

In the way described, we can reduce S$_d$ML to $\mathcal{H}(\downarrow)$. But, does the converse direction hold? 
First note that the following property is definable in $\mathcal{H}(\downarrow)$:

\begin{fact}\label{fact-hybrid-countsuccessor}
The property `there exist n successors of the current point' is definable in $\mathcal{H}(\downarrow)$.
\end{fact}

\begin{proof}
We prove it by building the desired formula. Let $n$ be a positive natural number. Consider the following $\mathcal{H}(\downarrow)$-formula:
\begin{center}
$\downarrow x(\Diamond\downarrow x_1(@_x\Diamond\downarrow x_2(...(@_x\Diamond\downarrow x_n(@_x\Box(\bigvee\limits_{0\le i\le n} x_i\land \bigwedge\limits_{1\le i<j\le n}\neg @_{x_i}x_j\underbrace{))...)))}_{n+2}$
\end{center}
The formula states that the current point $x$ has successors $x_1$, ..., $x_n$, that each node reachable from $x$ must be some $x_i$, where $1\le i\le n$, and that for any different $i$ and $j$ such that $1\le i,j\le n$, $x_i$ is distinct from $x_j$. Thus, there exist $n$ successors of the current point iff the stated hybrid formula holds at that point.
\end{proof}

But Example \ref{example-count} showed that this property is not definable in S$_d$ML. 

\begin{fact}\label{fact-comparehybrid}
$\mathcal{H}(\downarrow)$ is more expressive than S$_d$ML over models.
\end{fact}

Therefore S$_d$ML can be viewed as a fragment of $\mathcal{H}(\downarrow)$. Any hybrid logic at least as expressive as $\mathcal{H}(\downarrow)$ is more expressive than S$_d$ML. Even so, the hybrid translation described in Definition \ref{def-hybridtranslation} suggests that it may be viable to analyze validity in the logic S$_d$ML with expressive resources similar to those of $\mathcal{H}(\downarrow)$. 

\subsection{Digression on Recursion Axioms}\label{subsec-recursion}

One attractive format for axiomatizing logics of model change are recursion axioms in the style of dynamic-epistemic logic (see \cite{johandynamic}). As mentioned already, Boolean cases are available for S$_d$ML as well. We begin with the principle for $[-\;]$:  \footnote{Actually, the principle for $[-\;]$ is not necessary to show a complete set of recursion axioms, cf. \cite{lig}.}

\begin{fact}\label{fact-recursionBooleandynamic}Let $\varphi$, $\psi$ and $\chi$ be $\mathcal{L}_d$-formulas. Then it holds that
\begin{align}
[-\varphi][-\psi]\chi&\leftrightarrow\downarrow x[-\downarrow y(\varphi\lor@_x[-\varphi]@_y\psi)]\chi
\end{align}
where $x$ and $y$ are new nominal variables.
\end{fact}

\begin{proof}
Let $\lr{\mathcal{M},w}$ be a pointed model. We prove it by showing that $\mathcal{M}|_{\lr{w,\varphi}}|_{\lr{w,\psi}}$ and $\mathcal{M}|_{\lr{w,\downarrow y(\varphi\lor@_x[-\varphi]@_y\psi)}}$ are same, where $w\in V(x)$. Suppose not, then there must be some $v\in W$ such that $\lr{w,v}\in\mathcal{M}|_{\lr{w,\varphi}}|_{\lr{w,\psi}}$ and  $\lr{w,v}\not\in\mathcal{M}|_{\lr{w,\downarrow y(\varphi\lor@_x[-\varphi]@_y\psi)}}$, or that $\lr{w,v}\in\mathcal{M}|_{\lr{w,\downarrow y(\varphi\lor@_x[-\varphi]@_y\psi)}}$ and $\lr{w,v}\not\in\mathcal{M}|_{\lr{w,\varphi}}|_{\lr{w,\psi}}$.

Now consider the first case. From $\lr{w,v}\not\in\mathcal{M}|_{\lr{w,\downarrow y(\varphi\lor@_x[-\varphi]@_y\psi)}}$, we know that $\mathcal{M},v\vDash\varphi\lor@_x[-\varphi]@_y\psi$ where $v\in V(y)$. By $\lr{w,v}\in\mathcal{M}|_{\lr{w,\varphi}}|_{\lr{w,\psi}}$, it follows that $\mathcal{M}|_{\lr{w,\varphi}},v\not\vDash\psi$. Since $\mathcal{M}|_{\lr{w,\varphi}}|_{\lr{w,\psi}}$ is a submodel of  $\mathcal{M}|_{\lr{w,\varphi}}$, we obtain $\lr{w,v}\in\mathcal{M}|_{\lr{w,\varphi}}$. Consequently, it holds that $\mathcal{M},v\not\vDash\varphi$, thus, $\mathcal{M}|_{\lr{w,\varphi}},v\vDash\psi$. So we have arrived at a contradiction.

Next we consider the second case. By $\lr{w,v}\in\mathcal{M}|_{\lr{w,\downarrow y(\varphi\lor@_x[-\varphi]@_y\psi)}}$, it holds that $\mathcal{M},v\vDash\neg\varphi\land@_x[-\varphi]@_y\neg\psi$ where $v\in V(y)$. Then we know $\lr{w,v}\in\mathcal{M}|_{\lr{w,\varphi}}$. Besides, by $\lr{w,v}\not\in\mathcal{M}|_{\lr{w,\varphi}}|_{\lr{w,\psi}}$, we obtain $\mathcal{M}|_{\lr{w,\varphi}},v\not\vDash\psi$ that entails a contradiction.
\end{proof}

Note that some operators of $\mathcal{H}(\downarrow)$ occur in (6). From Definition \ref{def-hybridtranslation}, we know that it is equivalent with some formula of $\mathcal{H}(\downarrow)$. Consider formula $\downarrow x[-\downarrow y(\varphi\lor@_x[-\varphi]@_y\psi)]\chi$. By the semantics, that it is true at a pointed model $\lr{\mathcal{M},w}$ means that $w$ is $\chi$ in the model $\mathcal{M}|_{\lr{w,\downarrow y(\varphi\lor@_x[-\varphi]@_y\psi)}}$, where $V(x)=\{w\}$. Intuitively, the new model is obtained by removing all links from $w$ to the points that are $\varphi$, and to the points which are $\psi$ after removing the links from $w$ to $\varphi$-points. This is exactly what $[-\varphi][-\psi]\chi$ states. 

We now move to the case for $\Box$. It seems like that the following result will work:

\begin{fact}\label{fact-recursionBox} For each $[-\varphi]\Box\psi\in\mathcal{L}_d$, the following equivalence holds:
\begin{align}
[-\varphi]\Box\psi&\leftrightarrow\downarrow x\Box\downarrow y(\neg\varphi\to @_x[-\varphi]@_y\psi)
\end{align}
where $x$ and $y$ are new nominal variables.
\end{fact}

\begin{proof}
Let $\lr{\mathcal{M},w}$ be a pointed model. For the direction from left to right, we suppose that $\mathcal{M},w\vDash[-\varphi]\Box\psi$ and $\mathcal{M},w\not\vDash\downarrow x\Box\downarrow y(\neg\varphi\to @_x[-\varphi]@_y\psi)$. Then it holds that $w$ ($\in V(x)$) has a successor $v$ ($\in V(y)$) such that $\mathcal{M},v\vDash \neg\varphi\land @_x[-\varphi]@_y\neg\psi$. From $\mathcal{M},w\vDash[-\varphi]\Box\psi$, it follows that $\mathcal{M}|_{\lr{w,\varphi}},w\vDash\Box\psi$. Since $\mathcal{M},v\vDash\neg\varphi$, we obtain $\lr{w,v}\in\mathcal{M}|_{\lr{w,\varphi}}$. Thus it holds that $\mathcal{M}|_{\lr{w,\varphi}},v\vDash\psi$. Besides, $\mathcal{M},v\vDash\neg\varphi\land @_x[-\varphi]@_y\neg\psi$ entails $\mathcal{M},w\vDash[-\varphi]@_y\neg\psi$. Consequently, it holds that $\mathcal{M}|_{\lr{w,\varphi}},v\vDash\neg\psi$, which entails a contradiction.

For the converse direction, we assume that $\mathcal{M},w\vDash\downarrow x\Box\downarrow y(\neg\varphi\to @_x[-\varphi]@_y\psi)$ and $\mathcal{M},w\not\vDash[-\varphi]\Box\psi$. Then there exists $v\in W$ such that $\lr{w,v}\in R\setminus(\{w\}\times V(\varphi))$ and $\mathcal{M}|_{\lr{w,\varphi}},v\vDash\neg\psi$. Consider the case where $w$ and $v$ are named as $x$ and $y$ respectively. It holds that $\mathcal{M}|_{\lr{w,\varphi}},w\vDash@_y\neg\psi$. So we obtain $\mathcal{M}|_{\lr{w,\varphi}},w\vDash@_x[-\varphi]@_y\neg\psi$. Further more, from $\lr{w,v}\in R\setminus(\{w\}\times V(\varphi))$, we know $\lr{w,v}\in R$ and $\mathcal{M},v\vDash\neg\varphi$. Thus it is conclude that $\mathcal{M},w\not\vDash\downarrow x\Box\downarrow y(\neg\varphi\to @_x[-\varphi]@_y\psi)$.
\end{proof}

In formula (7), that $\downarrow x\Box\downarrow y(\neg\varphi\to @_x[-\varphi]@_y\psi)$ is true at $\lr{\mathcal{M},w}$ says that for each point $v$, if $v\in R(w)$ and $v$ is not $\varphi$, then $v$ is $\psi$ after deleting all links from $w$ to the $\varphi$-points. However, although formula (7) is valid, it is not the solution: the formula of the form $@_x[-\varphi]@_y\psi$ blocks the recursion format, even though we have that   

\begin{fact}\label{fact-recursion@}
For any $p\in\mathbf{P}$, $\mathcal{L}_d$-formulas $\varphi$, $\psi$ and $\chi$, and nominal variable $x$, the following equivalences hold:  
\begin{align}
[-\varphi]@_xp&\leftrightarrow@_xp \\
[-\varphi]@_x\neg\psi&\leftrightarrow\neg[-\varphi]@_x\psi\\
[-\varphi]@_x(\psi\land\chi)&\leftrightarrow[-\varphi]@_x\psi\land [-\varphi]@_x\chi\\
[-\varphi]@_x\Box\psi&\leftrightarrow\downarrow y@_x\Box\downarrow z(\neg(\varphi\land@_xy)\to@_y[-\varphi]@_z\psi)
\end{align}
where $y$ and $z$ are new nominal variables.
\end{fact}

\begin{proof}
The validity of (8)-(10) is straightforward. We now consider (11). Let $\lr{\mathcal{M},w}$ be a pointed model. From left to right. Suppose that $\mathcal{M},w\vDash[-\varphi]@_x\Box\psi$ and $\mathcal{M},w\not\vDash\downarrow y@_x\Box\downarrow z(\neg(\varphi\land@_xy)\to@_y[-\varphi]@_z\psi)$. Let $u$ be a point such that $V(x)=\{u\}$. Then it holds that $\mathcal{M},u\vDash\Diamond\downarrow z(\neg(\varphi\land@_xy)\land@_y[-\varphi]@_z\neg\psi)$ where $w\in V(y)$. Therefore there exists some point $v$ such that $Ruv$, $v\in V(z)$ and $\mathcal{M},v\vDash\neg(\varphi\land@_xy)\land@_y[-\varphi]@_z\neg\psi$. By $\mathcal{M},v\vDash\neg(\varphi\land@_xy)$, it holds that $\lr{u,v}\in\mathcal{M}|_{\lr{w,\varphi}}$. From $\mathcal{M},v\vDash@_y[-\varphi]@_z\neg\psi$, we obtain $\mathcal{M}|_{\lr{w,\varphi}},v\vDash\neg\psi$, which contradicts to $\mathcal{M},w\vDash[-\varphi]@_x\Box\psi$.

From right to left. Suppose that $\mathcal{M},w\vDash\downarrow y@_x\Box\downarrow z(\neg(\varphi\land@_xy)\to@_y[-\varphi]@_z\psi)$ and $\mathcal{M},w\not\vDash[-\varphi]@_x\Box\psi$. Let $u$ be a point such that $V(x)=\{u\}$. Then there exists some point $v$ such that  $\lr{u,v}\in \mathcal{M}|_{\lr{w,\varphi}}$ and $\mathcal{M}|_{\lr{w,\varphi}},v\vDash\neg\psi$. From $\mathcal{M},w\vDash\downarrow y@_x\Box\downarrow z(\neg(\varphi\land@_xy)\to@_y[-\varphi]@_z\psi)$, it holds that $\mathcal{M},v\vDash@_y[-\varphi]@_z\psi$ where $w\in V(y)$ and $v\in V(z)$. Consequently, we have $\mathcal{M}|_{\lr{w,\varphi}},v\vDash\psi$ that entails a contradiction.
\end{proof}

In the rest of this section, we are not going to present a  solution for this issue. Actually we conjecture that there exists no a recursion axiom for $[-\varphi]\Box\psi$ in $\mathcal{H}(\downarrow)$, which is contrasted with our initial intuition. However, given Corollary \ref{corollary}, there must be some sort of recursion axioms for it. Thus a question arises:

\medskip

\noindent\textbf{Open Problem.} Could there be a complete set of recursion axioms for S$_d$ML? 

\medskip

Through the above considerations, we understand why $\mathcal{H}(\downarrow)$ fails to do the job. In fact, there may be no easy solution, short of going to full first-order logic. All this suggests that, despite the axiomatizability in principle (as observed in Section \ref{sec-translation}), the structure of the logical validities of S$_d$ML is computationally complex. This suspicion will be confirmed in the next section, where we prove the undecidability of the logic.    

\section{Undecidability of S$_d$ML}\label{sec-undecidability}
Up to now, we have already shown that S$_d$ML is more expressive than the standard modal logic. Meanwhile, it is also a fragment of the hybrid logic $\mathcal{H}(\downarrow)$. It is well-known that the satisfiability problem for the standard modal logic is decidable. While, as noted in \cite{hybrid}, $\mathcal{H}(\downarrow)$ is undecidable. So, is S$_d$ML decidable or not? 

Actually, there are some fragments of $\mathcal{H}(\downarrow)$ that are decidable. For instance, \cite{hybriddecidable} shows that after removing all formulas containing a nesting of $\Box$, $\downarrow$ and $\Box$, $\mathcal{H}(\downarrow)$ becomes decidable. But in this section, we will present a negative answer to the question above, i.e., the satisfiability problem for S$_d$ML is undecidable. Moreover, we will identify the source of its high complexity. Before these results, we first show that S$_d$ML lacks both the tree model property and the finite model property.

\begin{theorem}\label{theorem-tree}
The logic S$_d$ML does not have the tree model property.
\end{theorem}

\begin{proof}
Consider the following formulas:
\begin{align*}
 (R_1)\qquad &p\land\Diamond p\land\Diamond\neg p\\
 (R_2)\qquad & \Box(p\to\Diamond p\land \Diamond\neg p)\\
 (R_3)\qquad & [-\neg p]\Box\Box p
\end{align*}
Let $\varphi_r:=(R_1\land R_2\land R_3)$. We now show that, for any $\mathcal{M}=\{W,R,V\}$ and $w\in W$, if $\mathcal{M},w\vDash \varphi_r$, then the evaluation point $w$ is reflexive. By $(R_1)$, $w$ has some $p$-successor(s) and some $\neg p$-successor(s). Formula $(R_2)$ states that each its $p$-successor $w_1$ also has at least one $p$-successor $w_2$ and at least one $\neg p$-successor $w_3$. From $(R_3)$ we know that, after deleting all links from $w$ to the $\neg p$-points, $w_1$ does not have $\neg p$-successors any longer. If node $w_1$ is not $w$, then $\varphi_r$ cannot be true at $w$.  That is to say, for each $v\in W$, if $Rwv$ and $\mathcal{M},v\vDash p$, then $v=w$, i.e., $R(w)\cap V(p)=\{w\}$. So if formula $\varphi_r$ is true, the evaluation point must be reflexive (with at least one $\neg p$-successor). A model for $\varphi_r$ is the $\mathcal{M}_2$ in the proof of Fact \ref{fact-bisimulation}, and $\varphi_r$ is true at the point $v_1$.  
\end{proof}

In addition, S$_d$ML also lacks the finite model property. To show this, inspired by the methods of \cite{hybrid}, we will construct a `spy point', i.e., a special point which has access in one step to any reachable point in the model. 

\begin{theorem}\label{theorem-finitemodel}
The logic S$_d$ML does not have the finite model property.
\end{theorem}

\begin{proof}
Let $\varphi_\infty$ be the conjunction of the following formulas:
\begin{align*}
(F_1) &&& s\land p\land\Box\neg s\land\Diamond p\land\Diamond\neg p\land\Box(\neg p\to\Box\bot)\\
(F_2)&&&\Box(p\to\Diamond s\land\Diamond \neg s\land\Box p)\\
(F_3)&&&\Box(p\to\Box(s\to \Box\neg s\land\Diamond \neg p))\\
(F_4)&&&[-\neg p]\Box\Box(s\to\neg\Diamond\neg p)\\ 
(F_5)&&&\Box(p\to\Box(\neg s\to\Diamond s\land\Diamond\neg s\land\Box p))\\
(F_6)&&&\Box(p\to\Box(\neg s\to\Box(s\to\Box\neg s\land\Diamond\neg p)))\\
(F_7)&&&[-\neg p]\Box\Box(\neg s\to\Box(s\to\neg\Diamond\neg p))\\
(\textit{Spy})&&&\Box(p\to\Box(\neg s\to[-\neg s]\Box\Diamond(p\land\Box s)))\\
(\textit{Irr})&&&\Box(p\to [-s]\Box\Diamond s)\\
(\textit{No-3cyc})&&&\neg\Diamond(p\land[-s]\Diamond[-s]\Diamond\Diamond(\neg s\land\Box\neg s))\\
(\textit{Trans})&&&\Box(p\to[-s]\Box\Box(\neg s\to[-\neg s]\Box\Diamond(\Box\neg s\land\Diamond\Box s)))
\end{align*}

First, we show that the formula $\varphi_\infty$ is satisfiable. Consider the following model  $\mathcal{M}$:
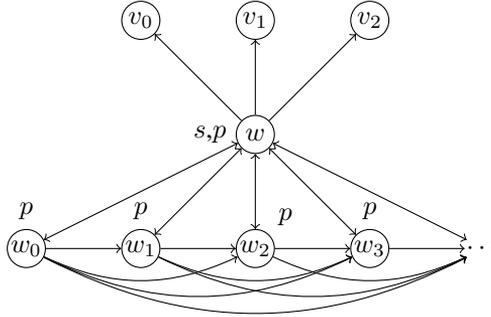
\begin{center}
\begin{tikzpicture}
\node(a)[circle,draw,inner sep=0pt,minimum size=5mm, label=left:{$s$,$p$}] at (0,0) {$w$};
\node(b)[circle,draw,inner sep=0pt,minimum size=5mm] at (-1.5,1.5) {$v_0$};
\node(c)[circle,draw,inner sep=0pt,minimum size=5mm] at (0,1.5) {$v_1$};
\node(d)[circle,draw,inner sep=0pt,minimum size=5mm] at (1.5,1.5) {$v_2$};
\node(e)[circle,draw,inner sep=0pt,minimum size=5mm,label=above:$p$] at (-3,-1.5) {$w_0$};
\node(f)[circle,draw,inner sep=0pt,minimum size=5mm,label=above:$p$] at (-1.5,-1.5) {$w_1$};
\node(g)[circle,draw,inner sep=0pt,minimum size=5mm,label=above right:$p$] at (0,-1.5) {$w_2$};
\node(h)[circle,draw,inner sep=0pt,minimum size=5mm,label=above:{$p$}] at (1.5,-1.5) {$w_3$};
\node(i)[circle,inner sep=0pt,minimum size=5mm] at (3,-1.5) {$\cdots$};
\draw[->](a) to (b);
\draw[->](a) to (c);
\draw[->](a) to (d);
\draw[<->](a) to (e);
\draw[<->](a) to (f);
\draw[<->](a) to (g);
\draw[<->](a) to (h);
\draw[<->](a) to (i);
\draw[->](e) to (f);
\draw[->](f) to (g);
\draw[->](g) to (h);
\draw[->](h) to (i);
\draw[->](e) to[bend right=25] (g);
\draw[->](f) to[bend right=25] (h);
\draw[->](f) to[bend right=26] (i);
\draw[->](e) to[bend right=26] (h);
\draw[->](e) to[bend right=27] (i);
\draw[->](g) to[bend right=25] (i);
\end{tikzpicture}  
\end{center} 
clearly, $\mathcal{M},w\vDash\varphi_\infty$. Thus there exists at least one model satisfying formula $\varphi_\infty$. 

Next, we show that for any $\mathcal{M}=\{W,R,V\}$ and $w\in W$, if $\mathcal{M},w\vDash\varphi_\infty$, then $W$ is infinite. For brevity, define that $B=\{v\in W| v\in R(w)\cap V(p)\}$, i.e., $B$ is the set of the $p$-successors of $w$. In the following proof, we assume that all previous conjuncts hold.

By $(F_1)$, the evaluation point $w$ is $(s\land p)$, and it cannot see any $s$-points. In particular, $w$ cannot see itself. Besides, $w$ has some $p$-successor(s) (i.e., $B\not=\emptyset$) and some $\neg p$-successor(s) (i.e., $R(w)\setminus B\not=\emptyset$). In addition, each point in $R(w)\setminus B$ is a dead end.

From formula $(F_2)$, we know that each element in $B$ can see some $(s\land p)$-point(s) and $(\neg s\land p)$-point(s), but cannot see any $\neg p$-points. Hence each point in $B$ has a successor distinct from itself.  

According to formula $(F_3)$, for any $w_1\in B$, each its $s$-successor can see some $\neg p$-point(s), but cannot see any $s$-points. 

By $(F_4)$, after removing all links from $w$ to $\neg p$-points, for each $w_1\in B$, each of its $s$-successors $w_2$ has no $\neg p$-successors. Thus $(F_4)$ shows that each $w_1\in B$ can see point $w$, and that for each $s$-point $w_2\in W$, if $w_2$ is a successor of $w_1$, then $w_2$ must be $w$. 

Formulas $(F_2)$-$(F_4)$ show the properties of the $(\neg s\land p)$-points which are accessible from the point $w$ in one step. Similarly, formulas $(F_5)$, $(F_6)$ and $(F_7)$ play the same role as $(F_2)$, $(F_3)$ and $(F_4)$ respectively, but focusing on showing the properties of the $(\neg s\land p)$-points that are accessible from $w$ in two steps. In particular, $(F_7)$ guarantees that every $(\neg s\land p)$-point $w_1$ which is accessible from $w$ in two steps can also see $w$, and that for each $s$-point $w_2\in W$, if $w_2$ is a successor of $w_1$, then $w_2$ must be $w$. 

Formula (\textit{Spy}) says that, for each $(\neg s\land p)$-point $w_1$ that is accessible from $w$ in two steps, after removing the links from $w_1$ to the $\neg s$-points, each successor $w_2$ of $w_1$ has a $p$-successor $w_3$ that only has $s$-successors. Besides, point $w_2$ must be $s$. By $(F_7)$, we know that $w_2=w$. In addition, by $(F_2)$, $w_3$ should have some $\neg s$-successor(s) if the cut induced by $[-s]$ does not take place at $w_3$. So it holds that $w_3=w_1$. In such a way, (\textit{Spy}) makes the evaluation point $w$ be a spy-point, and it ensures that each $(\neg s\land p)$-point $w_1$ which is accessible from $w$ in two steps is also accessible from $w$ in one step. 
By (\textit{Irr}), for each $w_1\in B$, after removing the link from $w_1$ to $w$, each its successor still can see $w$. Therefore each $w_1\in B$ is irreflexive. Besides, (\textit{No-3cyc}) disallows cycles of length 2 or 3 in $B$, and (\textit{Trans}) forces the accessibility relation $R$ to transitively order $B$.

Hence $B$ is an unbounded strict partial order, thus it is infinite and so is $W$. Now we have already shown that $\varphi_\infty$ is satisfiable, and that for each pointed model $\lr{\mathcal{M},w}$, if $\mathcal{M},w\vDash\varphi_\infty$, then $\mathcal{M}$ is an infinite model. This completes the proof.
\end{proof}

Now, by encoding the $N\times N$ tiling problem, we show that S$_d$ML is undecidable. A tile $t$ is a $1\times1$ square, of fixed orientation, with colored edges \textit{right($t$)}, \textit{left($t$)}, \textit{up($t$)} and \textit{down($t$)}. The $N\times N$ tiling problem is: given a finite set of tile types $T$, is there a function $f:N\times N\to T$ such that \textit{right(f(n,m))}=\textit{left(f(n+1,m))} and \textit{up(f(n,m))}=\textit{down(f(n,m+1))}? This problem is known to be undecidable (see \cite{tile}).

Following the ideas in \cite{hybrid}, we will use three modalities $\Diamond_s$, $\Diamond_u$ and $\Diamond_r$. Correspondingly, a model $\mathcal{M}=\{W,R_s,R_u,R_r,V\}$ now has three kinds of accessibility relations. We will construct a spy point over the relation $R_s$. The relations $R_u$ and $R_r$ represent moving up and to the right, respectively, from one tile to the other. Besides, the operator $[-\;]$ will work in the usual way, i.e., all of the three kinds of relations should be cut if the current point have some particular successors via them.\footnote{There is also no problem if we use three kinds of dynamic operators that correspond to the three kinds of accessibility relations respectively. In the proof of Theorem \ref{theorem-undecidable}, these three kinds of links are disjoint.} Let us see the details.

\begin{theorem}\label{theorem-undecidable}
The satisfiability problem for S$_d$ML is undecidable.
\end{theorem}

\begin{proof}
Let $T=\{T_1,...,T_n\}$ be a finite set of tile types. For each $T_i\in T$, we use $u(T_i)$, $d(T_i)$, $l(T_i)$, $r(T_i)$ to represent the colors of its up, down, left and right edges respectively. Besides, we code each tile type with a fixed propositional atom $t_i$. Now we will define a formula $\varphi_T$ such that $\varphi_T$ is satisfiable iff $T$ tiles $N\times N$. Consider the following formulas:
\begin{align*}
(M_1)&&& s\land p\land\Box_s\neg s\land\Diamond_s p\land\Diamond_s\neg p\land\Box_s(\neg p\to\Box_s\bot)\\
(M_2) &&& \Box_s(p\to\Diamond_s\top\land\Box_s(s\land\Diamond_s\neg p))\\
(M_3) &&& [-\neg p]\Box_s\Box_s(s\land\neg\Diamond_s\neg p)\\ 
(M_4)&&& \Box_s(p\to\Diamond_\dag\top\land\Box_\dag(\neg s\land p\land\Diamond_s\top\land\Box_s(s\land\Diamond_s\neg p)))& \dag\in\{u,r\}\\
(M_5)&&&[-\neg p]\Box_s\Box_\dag\Box_s\neg\Diamond_s\neg p& \dag\in\{u,r\}\\
(M_6)&&& \Box_s(p\to\Box_\dag(\Diamond_u\top\land\Diamond_r\top\land\Box_u(\neg s\land p)\land\Box_r(\neg s\land p)))&\dag\in\{u,r\}\\
(M_7)&&& \Box_s(p\to[-s]\Box_\dag(\Diamond_ss\land\neg\Diamond_\dag\neg\Diamond_ss))&\dag\in\{u,r\}\\
(\textit{Spy})&&&\Box_s(p\to\Box_\dag[-\neg s]\Box_s\Diamond_s(p\land\Box_u\bot\land\Box_r\bot))&\dag\in\{u,r\}\\
(\textit{Func})&&&\Box_s(p\to[-s]\Box_\dag[-\neg s]\Diamond_s\Diamond_s(p\land\neg\Diamond_ss\land\Diamond_\dag\top\land\\
&&&\Box_\dag(\Box_u\bot\land\Box_r\bot)) &\dag\in\{u,r\}\\
(\textit{No-UR})&&& \Box_s(p\to[-s]\Box_u\Box_r\Diamond_s s\land[-s]\Box_r\Box_u\Diamond_s s) \\
(\textit{No-URU})&&&\Box_s(p\to[-s]\Box_u\Box_r\Box_u\Diamond_s s)\\
(\textit{Conv})&&&\Box_s(p\to[-s]\Diamond_u[-s]\Diamond_r[-\neg s]\Diamond_s\Diamond_s(p\land\neg\Diamond_ss\land\\
&&&\Box_r(\Diamond_u\top\land\Diamond_r\top)\land\Diamond_u\neg\Diamond_ss\land\Diamond_r\Diamond_u(\Box_u\bot\land\Box_r\bot)))\\
(\textit{Unique})&&&\Box_s(p\to\bigvee\limits_{1\le i\le n}t_i\land\bigwedge\limits_{1\le i<j\le n}(t_i\to\neg t_j))\\
(\textit{Vert})&&& \Box_s(p\to \bigwedge\limits_{1\le i\le n}(t_i\to\Diamond_u\bigvee\limits_{1\le j\le n,\;u(T_i)=d(T_j)}t_j))\\
(\textit{Horiz})&&&\Box_s(p\to \bigwedge\limits_{1\le i\le n}(t_i\to\Diamond_r\bigvee\limits_{1\le j\le n,\;r(T_i)=l(T_j)}t_j))
\end{align*}

Define $\varphi_T$ as the conjunction of the formulas above. Let $\mathcal{M}=\{W,R_s,R_u,R_r,V\}$ be an arbitrary model and $w\in W$. We first analyze the effects of the stated formulas on what this model must look like.

More concretely, suppose that $\mathcal{M},w\vDash\varphi_T$. We show that $\mathcal{M}$ is a tiling of $N\times N$. For brevity, define $G:=\{v\in W|v\in R_s(w)\cap V(p)\}$ where $R_s(w)=\{v\in W|R_swv\}$, and we will use its elements to represent the tiles. In the following proof, we also assume that all previous conjuncts hold.

Formula $(M_1)$ is similar to $(F_1)$ occurring in the proof of Theorem \ref{theorem-finitemodel}, except that $(M_1)$ focuses on the relation $R_s$ only.

By $(M_2)$, each tile $w_1$ has some successor(s) via the relation $R_s$, and each such successor $w_2$ is $(s\land p)$ and also has some $(\neg s\land\neg p)$-successor(s) via $R_s$. It is worthy to note that formulas $(M_1)$ and $(M_2)$ illustrate that $R_s$ is irreflexive.

Formula $(M_3)$ ensures that each tile $w_1$ can see $w$ via $R_s$, and that for each $(s\land p)$-point $w_2\in W$, if $w_2$ is accessible from $w_1$ via $R_s$, then $w_2=w$.

$(M_4)$ states that each tile has some successor(s) via $R_u$ and some successor(s) via $R_r$. Besides, each point that is accessible from a tile via $R_u$ or $R_r$ is very similar to a tile: it is $(\neg s\land p)$, and has some $(s\land p)$-successor(s) $w_1$ via relation $R_s$ where each $w_1$ can see some $(\neg s\land\neg p)$-point(s) via $R_s$.  

By formula $(M_5)$, each $w_1\in W$ accessible from a tile via $R_u$ or $R_r$ can see $w$ by $R_s$. Also, for each $(s\land p)$-point $w_2\in W$, if it is accessible from $w_1$ via $R_s$, then $w_2=w$. 

Formula $(M_6)$ ensures that each $w_1\in W$ that is accessible from some tile via $R_u$ or $R_r$ also has some successor(s) via $R_u$ and some successor(s) via $R_r$. Besides, each its successor via $R_u$ or $R_r$ is $(\neg s\land p)$.

From formula $(M_7)$, it follows that both $R_u$ and $R_r$ are irreflexive and asymmetric.

By (\textit{Spy}), we know that the evaluation point $w$ is a spy point via the relation $R_s$.

Note that formula $(M_4)$ says that each tile has some tile(s) above it and some tile(s) to its right. Now, with (\textit{Func}), we have that each tile has exactly one tile above it and exactly one tile to its right.

By (\textit{No-UR}), no tile can be above/below as well as to the left/right of another tile. Formula (\textit{No-URU}) disallows cycles following successive steps of the $R_u$, $R_r$, and $R_u$ relations, in this order. Further more, (\textit{Conv}) ensures that the tiles are arranged as a grid. 

Formula (\textit{Unique}) guarantees that each tile has a unique type. (\textit{Vert}) and (\textit{Horiz}) force the colors of the tiles to match properly.

Thus we conclude that $\mathcal{M}$ is indeed a tiling of $N\times N$. 

\smallskip

Next we show the other direction required for our proof. Suppose the function $f: N\times N\to T$ is a tiling of $N\times N$. Define a model $\mathcal{M}=\{W,R_s,R_u,R_r,V\}$ as follows:
\begin{align*}
W&=(N\times N)\cup\{w,v\}\\
R_s&=\{\lr{w,v}\}\cup\{\lr{w,n}\mid n\in N\}\cup\{\lr{n,w}\mid n\in N\}\\
R_u&=\{\lr{\lr{n,m},\lr{n,m+1}}\mid n,m\in N\}\\ 
R_r&=\{\lr{\lr{n,m},\lr{n+1,m}}\mid n,m\in N\}\\
V(s)&=\{w\}\\
V(p)&=\{w\}\cup N\\
V(t_i)&=\{\lr{n,m}\in N\times N\mid f(\lr{n,m})=T_i\},\; {\rm{for\; each }}\; i\in\{1,...,n\}\\
V(q)&=\emptyset,\; {\rm{for\; any\; other\; propositional\; atoms}}\; q
\end{align*}
In particular, $w$ is a spy point in $\mathcal{M}$. By construction, we know that $\mathcal{M},w\vDash\varphi_T$.
\end{proof}

Thus, perhaps surprisingly, given the simple-looking syntax and semantics of S$_{d}$ML, the complexity of its logic is high. What is the reason for this high complexity, as contrasted with decidability of dynamic-epistemic logics of link deletion \cite{Liu2}? For SML, the reason offered by \cite{sabotage} is the stepwise nature of link deletion, and this is confirmed by the result in \cite{palstep} showing how a very simple stepwise variant of public announcement logic is undecidable. However, our case is different, since links are cut in a uniform definable way: the only remaining potential culprit is then the locality. 

To see the effects of this feature, recall the above formula (7). We already saw in Section \ref{subsec-recursion} that a formula of the form $@_x[-\varphi]@_y\psi$ blocks the recursion format. In contrast, consider a global version S$_d^g$ML of S$_d$ML. The truth condition for $[-\;]$ now reads:
\begin{align*}
\lr{W,R,V},w\vDash [-\varphi]\psi  \;\;{\rm{\textit{iff}}}\;\; \lr{W,R\setminus\{\lr{s,t}\in R\mid\mathcal{M},t\vDash\varphi\},V},w\vDash \psi.
\end{align*}
Given the global change made in this semantics, here is a valid recursion axiom for $\Box$: 
$$[-\varphi]\Box\psi\leftrightarrow\Box(\neg\varphi\to[-\varphi]\psi).$$
Indeed, following the general method for modal logics of definable model change presented in \cite{Liu2}, one can find a complete set of recursion axioms for S$_d^g$ML:

\begin{fact}\label{fact-globalversion}
The logic S$_d^g$ML is axiomatizable and decidable.
\end{fact}

The complexity effect of the local behavior of S$_d$ML also show at a crucial step in our proof of undecidability. In the proof of Theorem \ref{theorem-undecidable}, formula (\textit{Conv}) forces the tiles to satisfy a first-order convergence property, i.e., 

\vspace{-2mm}

$$\forall t\forall t_1\forall t_2(R_utt_1\land R_rt_1t_2\to\exists t_3(R_rtt_3\land R_ut_3t_2)).$$ 

\noindent As noted in \cite{openmind}, this property can give logics high complexity. 

By contrast,  convergence is not definable in S$_d^g$ML, even though we expand the model with some extra tools, e.g. a spy point. Roughly speaking, given two tiles $t_1$ and $t_2$ that have same properties, we still can distinguish between them with S$_d$ML, say, their properties will be different after cutting some links starting from $t_1$; however, we cannot do this with S$_d^g$ML, since links are cut in a global way.\footnote{From a technical point of view, to show that S$_d^g$ML cannot define the convergence property, we need its notion of bisimulation, which is easily defined.} 

The more general issue arising here goes beyond our specific logics of sabotage.

\medskip

\noindent\textbf{Open Problem.} Does making update operations local (world-relative) generate
undecidability in general for decidable dynamic-epistemic logics? 

\medskip

This would provide an alternative diagnosis to the comparison of sabotage and update offered in \cite{sabotage}, closer to the modified dynamic-epistemic logics studied in \cite{localpal}.

\section{Related Work}\label{sec-relatedwork}
This article is primarily inspired by existing work on sabotage games \cite{lig} and their applications. A good source for the latest developments in sabotage modal logics is \cite{sabotage}, which also has extensive references to current work on related modal logics for definable graph change. Meanwhile, a number of authors have studied other graph games using matching modal logics. For instance, in poison games, originating in graph theory, instead of deleting links, a player can poison a node, to make it inaccessible to her opponent. Poison games have been recently studied in the modal logics of \cite{poisonlogic}, using the close similarities between these systems and variants of so-called memory logics \cite{memory} in the hybrid tradition. In another tradition, that of Boolean network games, \cite{Declan} has proposed a logic of local fact change which can characterize Nash equilibria, providing a new way of looking at the interaction between graph games, network games and logics of control.


Throughout the paper, dynamic-epistemic logic \cite{baltag,johandynamic} has been used as a decidable contrasting design to our systems. Technically, our logic S$_{d}$ML has resemblances to several recent logics for local announcements. \cite{localpal} introduces a logic to characterize both global and local announcements. Similar to our set-up, it has definable updates of links, but there is also a difference. Although more expressive than public announcement logic, this logic is decidable. Moreover, we are inspired by other logics for local graph modifiers, too. For instance, \cite{changeoperator} investigates a special type of local SML, whose dynamic operator refers to a model transition that cuts a link from the current state and then evaluates a formula at the target of the deleted arrow. Finally, more akin to the above-mentioned \cite{Declan}, \cite{graphmodifier} studies local modifiers that update the valuation at the evaluation point, and shows that adding those modifiers dramatically increases the expressive power of the logic.   

Next, a highly relevant line of research for this article is hybrid logic, an area from which we have taken several basic techniques. As far as we know, \cite{hybrid} is the first to present the method of constructing a spy point, the main tool that was used to prove the undecidability of our logic S$_d$ML. \cite{Relation change} shows how relation-changing logics such as SML can be seen as fragments of hybrid logics, and identifies various decidable fragments of those logics with the help of hybrid translations. This fits with our findings in Section \ref{subsec-hybridtranslation}. Finally, \cite{hybridpal} merges hybrid logic with public announcement logic. Differently from the operator $[-]$ in S$_d$ML, the announcement modality there operates in a global way, making it possible to axiomatize the logic by means of recursion axioms.

It remains to note that this article fits with the general program recently proposed in \cite{graphgame} for a much broader study of analysis and design for graph games in tandem with matching modal logics. In particular, it proposed various meaningful new games, and identified general questions behind the match between logic and game.

\section{Summary and Further Directions}\label{sec-summary}
In this article, we started with a definable sabotage game S$_d$G that models some interesting phenomena in everyday life, and explored a matching logical system, definable sabotage modal logic S$_d$ML. We presented a first-order translation for the logic, showed a characterization theorem with regard to a novel notion of definable sabotage bisimulation, probed an axiomatization for S$_d$ML using recursion axioms in an extended hybrid language, and finally, we proved its undecidability.

Immediate technical open problems for our logic S$_d$ML resemble those in the literature for SML. For instance, we would like to have a good Hilbert-style proof theory, which may perhaps be found by analyzing semantic tableaux for S$_d$ML. Another open problem is the complexity of the schematic validities of our language.

Next, while our language can define winning positions for players in given finite graphs, it cannot express generic winning conditions across models. To obtain the latter, we need a modal $\mu$-calculus enriched with our local definable deletion modality, whose behavior shows the complexities already noted for sabotage $\mu$-calculus in \cite{sabotage}.

In terms of generality, one would like to establish the precise connections between our logic S$_d$M and other modal logics for graph games in the cited literature. For instance, the difference in expressive power that we noted in Section \ref{sec-bisimulationcharacterization} between S$_d$ML and SML does not preclude the existence of faithful embeddings either way.

As a final technical issue, we mentioned the contrast between locality and stepwise link deletion as sources of undecidability, discussed in Section \ref{sec-undecidability}. One could also merge these in a stepwise version of our logic, denoted  S$_d^s$ML. Clearly, its validities are different from those of S$_d$ML: for instance, $[-]$ is no longer self-dual. Our methods from Section 6 should also be able to prove its undecidability, but we have not yet been able to do so.

We end by stepping back to reality. In our introduction, we mentioned social networks \cite{Liu3}, where adding links (gaining friends or neighbors) is as important as deleting links (losing friends or neighbors). A connection between our logic and existing logics for social networks, and games played over these, would be a natural next step.

Another such step toward greater realism would arise when making connections to more elaborate versions of our game scenarios, for instance involving more complex independent goals for players than we have considered, or imperfect information when players cannot perfectly observe each other's moves. In general, such games may have probabilistic equilibria, and our logics would have to acquire interfaces with probability.

\section*{Acknowledgements}
I wish to thank Johan van Benthem and Fenrong Liu for their generous help and patient guidance  throughout all stages of this project. I would like to thank Alexandru Baltag for his inspiring suggestions. This work was supported by the Major Program of the National Social Science Foundation of China [17ZDA026].

\bibliographystyle{unsrt}
\bibliography{\jobname}

\end{document}